%% file: planar.tex
\documentclass{article}[11pt,a4paper]
\usepackage{epsf}
\usepackage{epsfig}
\usepackage{graphicx}
\usepackage{color}
\usepackage{amsmath}
\usepackage{mathrsfs}

\usepackage{amsthm}
\usepackage{latexsym}
\usepackage{amssymb}
\usepackage{amsmath}

\usepackage{stmaryrd}

\usepackage{epsf}
\usepackage{verbatim}

\newtheorem{theorem}{Theorem}[section]
\newtheorem{lemma}[theorem]{Lemma}

\newtheorem{definition}[theorem]{Definition}

\newcommand{\FF}{${\cal F}$-gate}
\newtheorem{myclaim}{Claim}

\begin{document}

\title{ {\bf Holographic Algorithms with Matchgates Capture Precisely Tractable Planar \#CSP}}

\vspace{0.3in}
\author{Jin-Yi Cai\footnotemark[1]
\and Pinyan Lu\footnotemark[2] \and Mingji Xia\footnotemark[3]}

\renewcommand{\thefootnote}{\fnsymbol{footnote}}

\footnotetext[1]{University of Wisconsin-Madison.
 {\tt jyc@cs.wisc.edu}. Supported by NSF CCF-0830488 and CCF-0511679.}

\footnotetext[2]{Microsoft Research Asia. {\tt
pinyanl@microsoft.com}}

\footnotetext[3]{Institute of Software, Chinese Academy of Sciences.
{\tt xmjljx@gmail.com}. Supported by the Grand Challenge Program
``Network Algorithms and Digital Information'' of the Institute of
Software, CAS. Partially supported by NSFC 60970003.}
\date{}
 \maketitle

\bibliographystyle{plain}

\begin{abstract}
Valiant introduced matchgate computation and holographic algorithms.
A number of seemingly exponential time
problems can be solved by this novel algorithmic paradigm in polynomial
time.
%
We show that, in a very strong sense,
matchgate computations and  holographic algorithms based on them
provide a universal methodology to a broad class of counting problems
studied in statistical physics community for decades.
They capture precisely those problems which are \#P-hard on general
graphs but computable in polynomial time on planar graphs.

More precisely,
we prove complexity dichotomy theorems in the framework of counting
CSP problems. The local constraint functions take Boolean
inputs, and can be arbitrary real-valued symmetric functions.
We prove  that, {\it every} problem in this class belongs to precisely
three  categories:
(1) those which are
tractable (i.e., polynomial time computable) on general graphs, or
(2) those which are \#P-hard on general graphs but tractable on planar
graphs, or
(3) those which are \#P-hard even on planar graphs.
The classification criteria are explicit.
Moreover, problems in category (2) are tractable on
planar graphs precisely by holographic algorithms with matchgates.
\end{abstract}


\input{intro}

\input{preliminaries}

\input{previous}

\input{interpolation}

\input{holant}

\input{csp}

\input{2-3_regular}

\input{physics.tex}

\section*{Acknowledgments}
We thank the following colleagues for their interests and helpful
comments: Xi Chen, Martin Dyer, Alan Frieze, Sean Hallgren, Leslie
Goldberg, Sorin Istrail, Richard Lipton, Jason Morton, Dana Randall,
Leslie Valiant and Santosh Vempala.

\bibliography{refs}





\end{document}

%% file: intro.tex
\section{Introduction}

Given a set of functions  $\mathcal{F}$, the Counting Constraint
Satisfaction Problem \#CSP($\mathcal{F}$) is the following
problem:  An input instance consists of
 a set of {\it variables}  $X= \{x_1, x_2,
\ldots, x_n\}$ and a set of {\it constraints} where each constraint
is a function $f \in \mathcal{F}$ applied to some variables in $X$.
The output is the sum, over all assignments to $X$, of the products
of these function evaluations. This sum-of-product evaluation
 is called the {\it partition function}.   In the special case where $f \in
\mathcal{F}$ outputs values in $\{0,1\}$ it counts the number of
satisfying assignments. But constraint functions taking real or complex values
are also interesting, called (real or complex) weighted
\#CSP. Our $\mathcal{F}$ consists of real or complex valued
functions in general. There is a
deeper reason for allowing this generality: The theory
of {\it holographic reductions} is a powerful tool which
operates naturally over ${\mathbb C}$, even  if the original problem has
only 0-1 valued functions.

A closely related framework for
locally constrained counting problems is called
Holant Problems~\cite{FOCS08,STOC09}.
This  framework is inspired by the introduction of
{\it Holographic Algorithms} by L.~Valiant~\cite{HA_J,AA_FOCS}.
In two ground-breaking papers~\cite{Valiant02,HA_J}
Valiant introduced matchgates and holographic algorithms
based on matchgates to
solve a number of
problems in polynomial time, which appear to require exponential time.
At the heart of these exotic algorithms is a tensor transformation
from a given problem to the problem of counting (complex) weighted
perfect matchings over planar graphs.  The latter problem has
a remarkable P-time algorithm
(FKT-algorithm)~\cite{TF1961,Kasteleyn1961,Kasteleyn1967}.
Planarity is crucial, as
counting perfect matchings over general graphs
is \#P-hard~\cite{Valiant79}.
Most of these holographic algorithms use
a suitable linear basis to {realize} locally a {\it  symmetric} function
with at most 3  Boolean variables on a  matchgate.
This work has been extended in \cite{STOC07}.
In particular we have obtained a complete characterization of
all realizable symmetric  functions
by matchgates over the complex field ${\mathbb C}$.

The study of ``tractable \#CSP'' type problems has a much longer
history in the statistical physics community (under different
names). Ever since Wilhelm Lenz who invented what is now known as
the Ising model, and asked his student Ernst
Ising~\cite{ising1925beitrag} to work on it, physicists have studied
so-called ``Exactly Solved
Models''~\cite{baxter1982exactly,mccoy1973two}. In the language of
modern complexity theory, physicists' notion of an ``Exactly
Solvable'' system corresponds
 to systems with polynomial time computable
partition functions. This is captured completely by
the computer science  notion of  ``tractable \#CSP''.
In Physics, many great researchers worked to build this intellectual
edifice, with remarkable contributions by Ising, Onsager, C.N.Yang,
T.D.Lee, Fisher, Temperley, Kasteleyn, Baxter, Lieb, Wilson
etc~\cite{ising1925beitrag,onsager1944crystal,yang1952spontaneous,yang1952statistical,lee1952statistical,TF1961,Kasteleyn1961,Kasteleyn1967,baxter1982exactly,lieb1981general,welsh1993complexity}.
A central question is to identify what
``systems'' can be solved ``exactly'' and what
``systems'' are ``difficult''.
The basic conclusion from physicists is that some ``systems'',
including the Ising model,
are ``exactly solvable'' for planar graphs, but
they appear difficult for higher dimensions.
There does not exist any rigorous or provable classification.
This is partly because the notion of a ``difficult''
partition function had no rigorous definition in physics.
However, in the language of complexity theory,
it is natural to consider the classification problem.
In this paper we do that, in the more
general setting of \#CSP with real valued constraint functions.
This will also shed light on why the valiant efforts by
physicists to generalize the  ``exactly solved'' planar system
to higher dimensions failed.
(In the appendix we will give some more background.)

Now turning from Physics to CS proper,
after Valiant introduced his holographic algorithms
with matchgates,  the following question can be raised:
Do these novel algorithms capture all P-time tractable counting problems
on planar graphs, {\it or} are there other more exotic
algorithmic paradigms yet undiscovered?
A suspicion (and perhaps an audacious proposition) is that
they  have indeed
captured all tractable planar counting
problems.
If so it would provide a universal methodology to a
broad class of counting problems
studied in statistical physics and beyond.
%
%
The results of this paper can be viewed
as an affirmation of that  suspicion.
Within the framework of weighted Boolean  \#CSP problems
our answer is YES, for {\it all} symmetric
real valued functions.
%

While  \#CSP problems provide a natural framework
to address this question, it turns out that the
deeper reason comes from Holant problems, which can be
described as follows: An input graph  $G=(V,E)$ is
given, where each $v \in V$ is attached a
 function $f_v \in \mathcal{F}$, mapping $\{0, 1\}^{\deg (v)}
\rightarrow {\mathbb R}$ or ${\mathbb C}$.
 We consider
all edge assignments $\sigma: E \rightarrow \{0,1\}$.
For each $\sigma$, $f_v$ takes its input bits
from the incident edges $E(v)$ at $v$, and evaluates
to $f_v(\sigma\mid_{E(v)})$.
The  counting problem on instance $G$ is to compute
${\rm Holant}_G = \sum_{\sigma}
\prod_{v\in V} f_v(\sigma\mid_{E(v)})$.
In effect, in a Holant problem, edges are variables and
vertices represent constraint functions.
This framework is very natural; e.g., the problem
of {\sc Perfect Matching} corresponds to attaching
the {\sc Exact-One} function at each vertex,
 taking 0-1 inputs.
The class of all Holant problems with
function set $\mathcal{F}$ is denoted by Holant($\mathcal{F}$).

Every \#CSP problem can be simulated by a Holant problem.
Represent  any instance of a \#CSP problem by  a bipartite graph where
LHS are labeled  by variables and RHS are labeled by constraints.
Denote by $=_k: \{0, 1\}^k \rightarrow \{0, 1\}$
the {\sc Equality} function of arity $k$, which is 1 on $0^k$ and $1^k$,
and is 0 elsewhere.
Then we can turn the  \#CSP instance to
an input graph
of
a Holant problem, by  replacing every variable  vertex $v$ on LHS
by $=_{\deg(v)}$.
In fact,  \#CSP($\mathcal{F}$)  is exactly the same as
Holant($\mathcal{F} \cup \{=_k \mid k \ge 1\}$).
Thus,  \#CSP problems can be viewed as Holant problems
where all {\sc Equality} functions are available for free,
or assumed to be present.
However, when we wish to discuss some restricted classes
of counting problems, e.g., for 3-regular  graphs,
the framework of Holant problems is the more natural one.
And as it turns out, the main technical breakthrough for
 our dichotomy theorem for planar \#CSP comes from Holant problems.

In this paper we will only consider Boolean variables $X$.
For a symmetric function on $k$ variables,
we denote it as $[f_0, f_1, \ldots, f_k]$,
where $f_i$ is the value of $f$ on inputs of  Hamming weight $i$.
E.g., $(=_1) = [1,1], (=_2) = [1,0,1]$ and $(=_3) = [1,0,0,1]$ etc.
When we relax Holant problems by allowing all
{\sc Equality} functions for free, we obtain  \#CSP.
We can also consider other relaxations.
Let ${\bf 0} = [1,0]$ and ${\bf 1} = [0,1]$ denote
the constant 0 and 1 unary (arity 1) functions.
Then Holant$^c$ is the natural class
of Holant problems where ${\bf 0}$ and ${\bf 1}$
are free.  This amounts to
computing Holant on input graphs where we can
set  0 or 1 to some dangling edges (one end has degree 1).
Another class of Holant problems is called  Holant$^*$ problems
 where we assume
all unary functions $[u_0, u_1]$ are free.

In \cite{STOC09} we obtained a dichotomy theorem
for (complex) Holant$^*$ problems and (real)  Holant$^c$
problems. The  dichotomy criterion
for  Holant$^*$ problems is still valid for {\it planar graphs}.
The proof of dichotomy theorems in this paper starts from there.

In Section~\ref{section:holant}, we
prove that for any real-valued symmetric function set $\mathcal{F}$,
the  planar  Holant$^c(\mathcal{F})$ problem
is tractable (i.e., computable in P) iff either it is tractable
over general graphs (for which we already have an effective
dichotomy theorem~\cite{STOC09}), or
 it is tractable because every function in $\mathcal{F}$
is realizable by a matchgate, in which case
  the  planar  Holant$^c(\mathcal{F})$ problem
 is computable by matchgates in P-time using FKT.
In {\it all other cases} the problem is \#P-hard.\footnote{Strictly
speaking, we must only consider $\mathcal{F}$ where functions take
computable  real numbers; this will be assumed implicitly.} A
crucial ingredient of the proof is a crossover construction whose
validity is proved algebraically, which seems to defy any direct
combinatorial justification.

Our second theorem (Section~\ref{section:csp})
is about planar  \#CSP problems.
We prove that for any set of
 real-valued symmetric functions $\mathcal{F}$,
the  planar   \#CSP$(\mathcal{F})$ problem
is  tractable iff either it is tractable
as  \#CSP$(\mathcal{F})$ without the planarity
restriction (for which we  have an effective
dichotomy theorem~\cite{STOC09}), or
 it is tractable because every function in $\mathcal{F}$
is realizable by a matchgate under a specific holographic
basis transformation.
Thus planar \#CSP$(\mathcal{F})$ is solvable by a holographic
algorithm in the second case.  For all other  $\mathcal{F}$
the problem is \#P-hard.
The proof of this dichotomy theorem for planar \#CSP
is built on the one for  planar Holant$^c$ in Section~\ref{section:holant}.

Our third result is a dichotomy theorem for
planar 2-3 regular bipartite
Holant problems (Section~\ref{section:2-3-regular}).
(This theorem deals with Holant problems without
assuming unary  ${\bf 0}$ and ${\bf 1}$.)
This includes Holant problems for 3-regular
graphs as a special case.
The tractability criterion is the same:
Either it is  tractable for general graphs (for which we also
have an effective
dichotomy theorem~\cite{CHL09}),
or it is tractable by a suitable holographic
algorithm, which is a  holographic
reduction  to FKT using matchgates.
In all other cases the problem is \#P-hard.

The  three dichotomy theorems are not mutually subsumed by
each other and are of independent interest.  In each framework
the respective theorem is a demonstration
that holographic algorithms with matchgates capture
precisely those \#P-hard problems which become tractable
for planar graphs.

%% file: preliminaries.tex
\section{Preliminaries}\label{sec:background}

\subsection{Problem and Definitions}
Our functions take values in $\mathbb{C}$ by default.
The framework of
 Holant problems is defined for functions
mapping any  $[q]^k\rightarrow \mathbb{C}$ for a finite $q$. Our
results in this paper are for the Boolean case $q=2$.
So we give the following definitions only
for $q=2$ for notational simplicity.

A {\it signature grid} $\Omega = (H, {\mathcal F}, \pi)$
consists of a graph $H=(V,E)$, and a labeling $\pi$
which labels  each vertex with
a function $f_v \in {\mathcal F}$.
 The Holant problem on instance
$\Omega$ is to compute ${\rm Holant}_\Omega=\sum_{\sigma}
\prod_{v\in V} f_v(\sigma\mid_{E(v)})$, a sum over all edge
assignments $\sigma: E \rightarrow \{0,1\}$.
A function $f_v$ can be represented as a vector of length $2^{\deg(v)}$,
or a tensor in $({\mathbb{C}}^{2})^{\otimes \deg(v)}$.
A function $f \in {\mathcal F}$ is also called a {\it signature}.
We denote by $=_k$ the {\sc Equality} signature of arity $k$.
A symmetric function $f$ on $k$ Boolean variables
can be expressed by $[f_0,f_1,\ldots,f_k]$, where $f_i$ is the value
of $f$ on inputs of Hamming weight $i$.
Thus, $(=_k)=[1,0,\ldots,0,1]$ (with $k-1$ zeros).
A Holant problem is parameterized by a set of signatures.
\begin{definition}
Given a set of signatures ${\mathcal F}$, we define a counting problem
${\rm Holant}({\mathcal F})$:

Input: A {\it signature grid} $\Omega = (G, {\mathcal F}, \pi)$;

Output: ${\rm Holant}_\Omega$.
\end{definition}

Planar Holant problems are Holant problems on planar graphs.
\begin{definition}
Given a set of signatures ${\mathcal F}$, we define a counting problem
{\rm Pl-Holant}$({\mathcal F})$:

Input: A {\it signature grid} $\Omega = (G, {\mathcal F}, \pi)$, where $G$ is a planar graph;

Output: ${\rm Holant}_\Omega$.
\end{definition}

We would like to characterize the complexity of Holant problems in
terms of its signature sets.~\footnote{Usually our set of
signatures ${\mathcal F}$ is a finite set, and the assertion
of either ${\rm Holant}({\mathcal F})$ is tractable
or \#P-hard has the usual meaning.  However our dichotomy
theorem is actually stronger: we allow ${\mathcal F}$
to be infinite, e.g., to include $\{=_1, =_2, =_3,\ldots\}$
or all unary signatures.
${\rm Holant}({\mathcal F})$ is tractable means that
it is computable in P
even when we include the description
of the signatures in the input $\Omega$ in  the input size.
${\rm Holant}({\mathcal F})$ is   \#P-hard  means that
there exists a finite subset of ${\mathcal F}$
for which the problem is  \#P-hard.}
For some  ${\mathcal F}$, it is possible that
Holant$({\mathcal F})$ is \#P-hard, while {\rm Pl-Holant}$({\mathcal F})$ is tractable.
These new tractable cases make dichotomies for planar Holant problems more
challenging.
This is also the focus of this work.
Some special families of Holant
problems have already been widely studied. For
example, if ${\mathcal F}$ contains all {\sc Equality} signatures
$\{=_1, =_2, =_3,\ldots\}$, then this is exactly the weighted \#CSP
problem.  Pl-\#CSP denotes the restriction  of \#CSP
to  planar structures, i.e., the standard bipartite
graphs representing the input instances of  \#CSP are planar.
   In \cite{STOC09}, we also introduced the
following two special families of Holant problems
by assuming some signatures are freely available.

\begin{definition}
Let ${\mathcal U}$ denote the set of all unary signatures. Given a set
of signatures ${\mathcal F}$, we use ${\rm Holant}^* ({\mathcal F})$ (or {\rm Pl-Holant}$^* ({\mathcal F})$ respectively) to
denote ${\rm Holant}({\mathcal F}\cup {\mathcal U})$ (or {\rm Pl-Holant}$({\mathcal F}\cup {\mathcal U})$ respectively).
\end{definition}

\begin{definition}
Given a set of signatures ${\mathcal F}$, we use
${\rm Holant}^c
({\mathcal F})$ (or {\rm Pl-Holant}$^c ({\mathcal F})$ respectively) to denote ${\rm Holant}({\mathcal F}\cup \{ [1,0],
[0,1] \})$ ( or {\rm Pl-Holant}$({\mathcal F}\cup \{ [1,0],
[0,1] \})$ respectively).
\end{definition}

Replacing a signature $f\in {\mathcal F}$ by a constant multiple $cf$, where $c
\neq 0$, does not change the complexity of ${\rm Holant}({\mathcal
F})$. So we  view $f$ and $cf$ as the same signature.
An important property of a signature is whether  it is degenerate.
\begin{definition}
A signature is degenerate iff it is a
tensor product of unary signatures. In particular, a symmetric signature in ${\mathcal F}$
is degenerate iff it can be expressed as
$\lambda[x,y]^{\otimes k}$.
\end{definition}

\subsection{${\cal F}$-Gate and Matchgate}
A signature from ${\cal F}$
is a basic function which can be used at a vertex in an input graph.
 Instead of a single vertex, we can use graph fragments to generalize
this notion. An ${\cal F}$-gate $\Gamma$ is a tuple $(H, {\cal F},
\pi)$, where $H = (V,E,D)$ is a graph where the edge set consists of
regular edges $E$ and dangling edges $D$. Some nodes of degree 1 are
designated as external nodes, and all other nodes are internal
nodes;  a dangling edge connects an internal node to an external
node, while a regular edge connects two internal nodes. The labeling
$\pi$ assigns a function from ${\cal F}$ to each internal node. The
dangling edges define variables for the ${\cal F}$-gate. (See Figure
\ref{Figs:F-gate} for one example.)
\begin{figure}[htbp]
\begin{center}
\includegraphics[width=2 in]{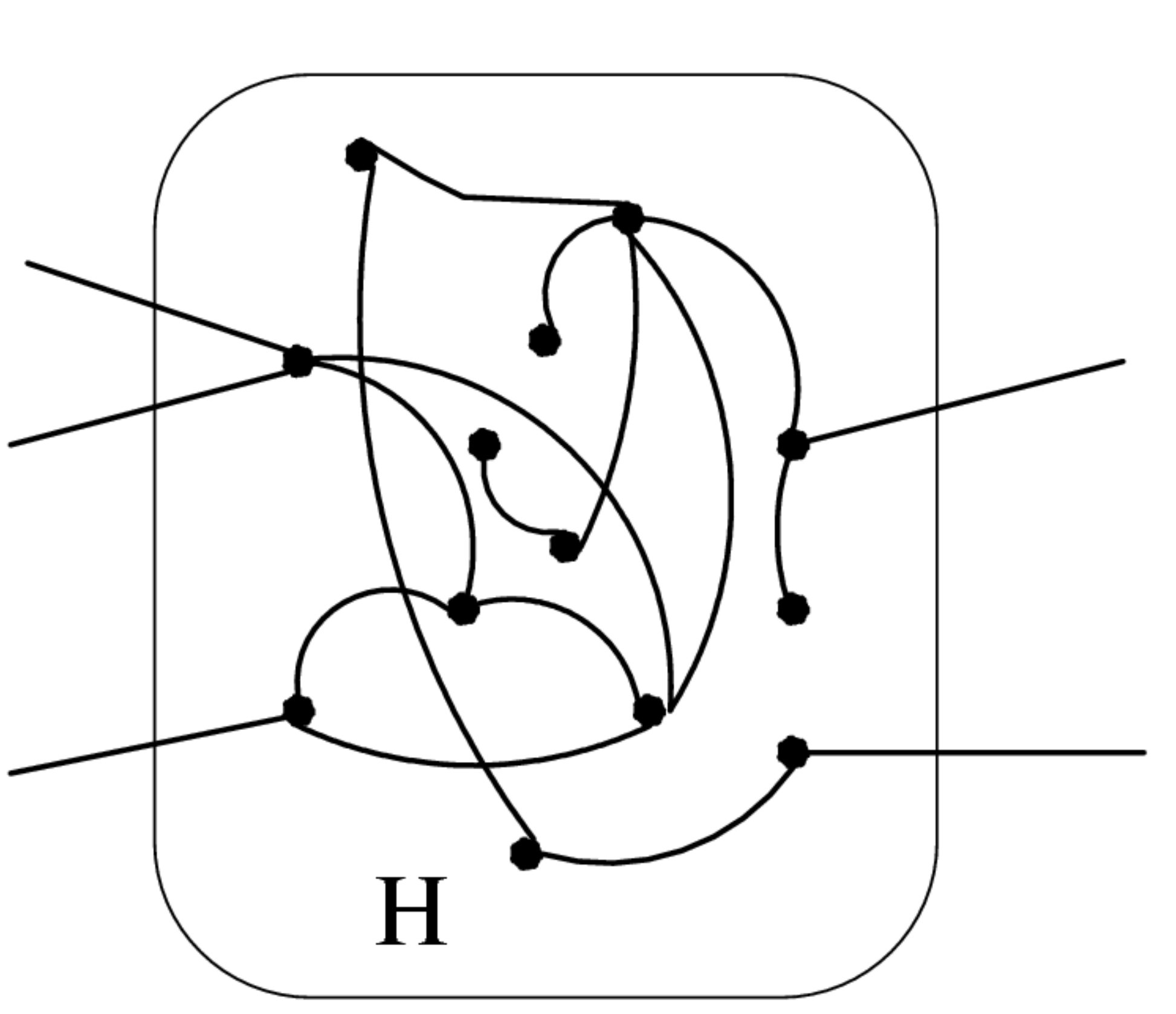} \caption{An ${\cal F}$-gate
with 5 dangling edges.} \label{Figs:F-gate}
\end{center}
\end{figure}
We denote the regular edges in $E$ by
$1,2,\ldots,m$, and  denote the dangling edges in $D$ by
$m+1,m+2,\ldots, m+n$. Then we can define a function for this ${\cal
F}$-gate $\Gamma = (H, {\cal F}, \pi)$,
\[\Gamma(y_1, y_2,
\ldots, y_n)=\sum_{x_1, x_2, \ldots x_m}H(x_1,x_2,\ldots,x_m,y_1,y_2,
\ldots y_n),\] where $(y_1, y_2, \ldots, y_n) \in \{0,1\}^n$ denotes
an assignment on the dangling edges
 and $H(x_1,x_2,\ldots, x_m,y_1,y_2,\ldots,y_n)$ denotes
the value of the signature grid on an assignment of all edges. We
will also call this function the signature of the ${\cal F}$-gate
$\Gamma$. An ${\cal F}$-gate can be used in a signature grid as if
it is just a single node with the particular signature.

Using the idea of ${\cal F}$-gates, we can reduce one Holant
problem to another. Let $g$ be the signature of some ${\cal
F}$-gate $\Gamma$. Then ${\rm Holant} ({\cal F} \cup \{g
\})\leq_T {\rm Holant} ({\cal F}) $. The reduction is quite
simple. Given an instance of ${\rm Holant} ({\cal F} \cup \{g
\})$, by replacing every appearance of $g$ by an ${\cal F}$-gate
$\Gamma$, we get an instance of ${\rm Holant} ({\cal F}) $. Since
the signature of  $\Gamma$ is $g$, the values for these two
signature grids are identical.

We note that even for a very simple signature set ${\cal F}$, the
signatures for all ${\cal F}$-gates can be quite complicated and
expressive. Matchgate signatures are an example. Matchgate is
introduced by Valiant $\cite{Valiant02,Valiant02_matchgate,HA_J}$,
whose definition is combinatorial in nature. Matchgates can be
viewed as a special case of planar ${\cal F}$-gates, where ${\cal
F}$ contains Exact-One functions of all arities and weight functions
($[1,0,w]$, $w\in \mathbb{C}$) on edges. (Formally, we replace each
matchgate edge of weight $w$ by a path of length 2, and the new node
on the path is assigned weight function $[1, 0, w]$.) The signature
function $\Gamma$ defined above for a matchgate is called a
matchgate signature, or a standard signature. A signature function
is realizable by a matchgate if it is the  standard signature of
that matchgate. (After  a  holographic transformation, a signature
function is realizable under a basis if it is the transformed
signature of a matchgate; see below.)

\subsection{Holographic Reduction}

To introduce the idea of holographic reductions, it is convenient
to consider bipartite graphs.  This is without loss of generality.
For any general graph, we can make it
bipartite by replacing each edge by a path of length two,
and giving each
new vertex  the {\sc Equality} function  $=_2$ on 2 inputs.
(This is just the incident graph.)

We use ${\rm Holant}( {\mathcal G}|{\mathcal R})$ to denote all counting
problems, expressed as Holant problems on bipartite graphs
$H=(U,V,E)$, where each signature for a vertex in $U$ or $V$ is from
${\mathcal G}$ or ${\mathcal R}$, respectively.  An input instance
for the bipartite
 Holant problem is a bipartite signature grid and is denoted as $\Omega =
(H, {\mathcal G}|{\mathcal R},\pi)$. Signatures in ${\mathcal G}$ are
denoted by column vectors (or
contravariant tensors); signatures in ${\mathcal R}$ are
denoted by row vectors (or covariant
tensors)~\cite{dodson}.

One can perform (contravariant and covariant) tensor
transformations on the signatures.
We will define a simple version of  holographic reductions,
which are invertible.
They are called holographic because they may
produce exponential cancellations in the tensor space.
Suppose  ${\rm Holant}( {\mathcal  G}|{\mathcal R})$ and
 ${\rm Holant}( {\mathcal G'}|{\mathcal R'})$ are two Holant problems defined for the same
family of graphs, and
 $T \in {\bf GL}_2({\mathbb C})$.
We say that there is an (invertible) holographic
reduction from  ${\rm Holant}( {\mathcal G}|{\mathcal R})$ to
 ${\rm Holant}( {\mathcal G'}|{\mathcal R'})$, and $T$ is the basis
transformation, if
the {\it contravariant} transformation
$G' = T^{\otimes g} G$ and the {\it covariant} transformation
$R=R' T^{\otimes r}$ map $G\in {\mathcal G}$ to $G'\in {\mathcal G'}$
and $R\in {\mathcal R}$  to $R' \in {\mathcal R'}$,
and vice versa,
where $G$ and $R$ have arity $g$ and $r$ respectively.
(Notice the reversal of directions when the
transformation $T^{\otimes n}$ is applied. This is the meaning of
{\it contravariance} and {\it covariance}.)

\begin{theorem}[Valiant's Holant Theorem \cite{HA_J}]\label{thm:holant}
Suppose there is  a holographic reduction from $\#{\mathcal G}|{\mathcal
R}$ to $\#{\mathcal G'}|{\mathcal R'}$ mapping signature grid $\Omega$
to $\Omega'$, then ${\rm Holant}_{\Omega} = {\rm Holant}_{\Omega'}.$
\end{theorem}

In particular,  for invertible  holographic reductions from ${\rm Holant}( {\mathcal G}|{\mathcal R})$ to
 ${\rm Holant}( {\mathcal G'}|{\mathcal R'})$, one problem
is in P iff the other one is, and similarly one problem is \#P-hard
iff the other one is also.

In the study of Holant problems, we will commonly transfer between bipartite and
non-bipartite settings. When this does not cause confusion,
 we do not
distinguish signatures between  column vectors (or
contravariant tensors) and  row vectors (or covariant
tensors).  Whenever we write a transformation as $T^{\otimes n} F$ or $T \mathcal{F}$,
we view the signature or signatures as  column vectors (or
contravariant tensors);
  whenever we write a transformation as $F T^{\otimes n} $ or $\mathcal{F} T $,
we view the signature or signatures as  row vectors (or covariant
tensors).

%% file: previous.tex

\subsection{Some Known Dichotomy Results}
In this subsection, we state some known dichotomy theorems.
We first review three dichotomy theorems from \cite{STOC09}.

\begin{theorem}\label{thm-holant-star}
Let $\mathcal{F}$ be a set of symmetric signatures over $\mathbb{C}$.
Then $\rm{Holant}^*(\mathcal{F})$ is computable in polynomial time
in the following three cases. In all other cases,
$\rm{Holant}^*(\mathcal{F})$ is \#P-hard.
\begin{enumerate}
    \item Every signature in $\mathcal{F}$ is of arity no more than two;
    \item There exist two constants $a$ and $b$ (not both zero, depending
        only on $\mathcal{F}$), such that for every signature
        $[x_0,x_1,\ldots,x_n] \in \mathcal{F}$ one of the two
 conditions is
        satisfied: (1) for every $k=0,1,\ldots,n-2$, we have
        $ax_{k}+bx_{k+1}-ax_{k+2}=0$; (2) $n=2$ and the signature
        $[x_0,x_1,x_2]$ is of form $[2a\lambda,b\lambda,-2a\lambda]$.
    \item For every signature $[x_0,x_1,\ldots,x_n] \in \mathcal{F}$, one
        of the two conditions is satisfied: (1) For every $k=0,1,\ldots,n-2$,
        we have $x_{k}+x_{k+2}=0$; (2) $n=2$ and the signature $[x_0,x_1,x_2]$
        is of form $[\lambda,0,\lambda]$.
\end{enumerate}
The same dichotomy also holds for Pl-Holant$^*(\mathcal{F})$.
\end{theorem}

\begin{theorem}\label{thm-holnat-c}
Let $\mathcal{F}$ be a set of \emph{real} symmetric signatures, and let
$\mathcal{F}_1, \mathcal{F}_2$ and $\mathcal{F}_3$
 be three families of signatures
defined as
\begin{eqnarray*}
    \mathcal{F}_1 & = &
        \{\lambda([1,0]^{\otimes k}+i^r[0,1]^{\otimes k}) |
            \lambda \in \mathbb{C}, k=1,2, \ldots, r=0,1,2,3\}; \\
    \mathcal{F}_2 & = &
        \{\lambda([1,1]^{\otimes k}+i^r[1,-1]^{\otimes k}) |
            \lambda \in \mathbb{C}, k=1,2, \ldots, r=0,1,2,3\}; \\
    \mathcal{F}_3 & = &
        \{\lambda([1,i]^{\otimes k}+i^r[1,-i]^{\otimes k}) |
            \lambda \in \mathbb{C}, k=1,2, \ldots, r=0,1,2,3\}.
\end{eqnarray*}
Then ${\rm Holant}^c (\mathcal{F})$ is computable in polynomial time if
(1) After removing unary signatures from $\mathcal{F}$,
it falls in one of the three Classes of Theorem~\ref{thm-holant-star}
(this implies  ${\rm Holant}^* (\mathcal{F})$ is computable in polynomial time)
 or
(2) (Without removing any unary signature)
$\mathcal{F} \subseteq \mathcal{F}_1 \cup \mathcal{F}_2 \cup \mathcal{F}_3$.
 Otherwise, ${\rm Holant}^c (\mathcal{F})$ is \#P-hard.
\end{theorem}

Here we explicitly list all the real signatures in $\mathcal{F}_1
\cup \mathcal{F}_2 \cup \mathcal{F}_3$, up to an arbitrary scalar
factor:
\begin{enumerate}
  \item ($\mathcal{F}_1$):  $[1,0,0,\ldots, 0, 1 \mbox{ }(\mbox{or } -1)],$
  \item ($\mathcal{F}_2$):  $[1,0,1,0,\ldots, 0 \mbox{ }(\mbox{or } 1 )],$
  \item ($\mathcal{F}_2$):  $[0,1,0,1,\ldots, 0 \mbox{ }(\mbox{or } 1 )],$
  \item ($\mathcal{F}_3$):  $[1,0,-1,0,1,0,-1,0,\ldots, 0 \mbox{ }(\mbox{or } 1 \mbox{ or } -1)],$
  \item ($\mathcal{F}_3$):  $[0,1,0,-1,0,1,0,-1,\ldots, 0 \mbox{ }(\mbox{or } 1 \mbox{ or } -1)],$
  \item ($\mathcal{F}_3$):  $[1,1,-1,-1,1,1,-1,-1,\ldots, 1 \mbox{ }(\mbox{or } -1)],$
  \item ($\mathcal{F}_3$):  $[1,-1,-1,1,1,-1,-1,1,\ldots, 1 \mbox{ }(\mbox{or } -1)].$
\end{enumerate}

\begin{definition}
A $k$-ary function $f(x_1,\ldots,x_k)$ is affine if it
has the form
\begin{displaymath}
\chi_{[AX=0]}i^{\sum_{j=1}^{n}\langle \alpha_j,X\rangle}
\end{displaymath}
where $X=(x_1,x_2,\ldots,x_k,1)$, and $\chi$ is a 0-1 indicator function
such that $\chi_{[AX=0]}$ is 1 iff $AX=0$. Note that the inner
product $\langle \alpha,X\rangle$
 is calculated over $\mathbb{F}_2$, while the
summation over $j$ on the exponent of $i = \sqrt{-1}$
is over $\mathbb{F}_4$. We use $\mathcal{A}$
to denote the set of all affine functions.

We use $\mathcal{P}$
to denote the set of functions which can be expressed as a
product of unary functions, binary equality functions  ($[1,0,1]$
on some two variables)
and binary disequality functions ($[0,1,0]$ on some two variables).
\end{definition}

\begin{theorem}\label{thm:dichotomy}
Suppose $\mathcal{F}$ is a set of functions mapping Boolean inputs to
complex numbers. If $\mathcal{F} \subseteq \mathcal{A}$ or $\mathcal{F}
\subseteq \mathcal{P}$, then \#CSP($\mathcal{F}$) is computable in
polynomial time.
Otherwise, \#CSP($\mathcal{F}$) is \#P-hard.
\end{theorem}

As we mentioned in \cite{STOC09}, the class $\mathcal{A}$ is a natural generalization of
the symmetric signatures family $\mathcal{F}_1 \cup \mathcal{F}_2 \cup \mathcal{F}_3$.
It is easy to show that the set of symmetric signatures in  $\mathcal{A}$ is exactly
$\mathcal{F}_1 \cup \mathcal{F}_2 \cup \mathcal{F}_3$.

\vspace{.1in}

The following dichotomy for 2-3 regular graphs is from \cite{cai}.

\begin{theorem}\label{lemma-cai} {\rm (\cite{cai})}
The problem Holant$([y_0, y_1, y_2]|[1,0,0,1])$ is \#P-hard for all
$y_0, y_1, y_2 \in \mathbb{C}$ except in the following cases, for which
the problem is in P: (1) $y_1^2=y_0 y_2$;
(2) $y_0^{12}=y_1^{12}$ and $y_0 y_2=-y_1^2$ ( $y_1 \neq 0$) ;
(3) $y_1=0$;
(4) $y_0=y_2=0$.
If we restrict the input to planar graphs, then these four categories are tractable in P,
as well as a fifth category $y_0^3=y_2^3$, and the problem remains \#P-hard in all other cases.
\end{theorem}

\subsection{Characterization of Realizable Signatures by Matchgates}
A matchgate is called even (respectively odd)
if it has an even  (respectively odd) number of vertices.
The following two lemmas are  from \cite{MGI}.

\begin{lemma}
\label{even-geometric-series-lemma}
A symmetric signature
$[z_0, \ldots, z_m]$ is the standard signature of some
even matchgate iff  for all odd $i$, $z_i =0$,
and there exist $r_1$ and $r_2$ not both zero,
 such that for every even $2 \le k \le m$,
\[
r_1 z_{k-2} = r_2 z_k.
\]
\end{lemma}

\begin{lemma}
\label{odd-geometric-series-lemma}
A symmetric signature
$[z_0, \ldots, z_m]$ is the standard signature of some
odd matchgate  iff  for all even $i$,
$z_i =0$, and there exist $r_1$ and $r_2$ not both zero,
 such that for every odd $3 \le k \le m$,
\[
r_1 z_{k-2} = r_2 z_k.
\]
\end{lemma}

In  \cite{STACS07}, we characterized
 all symmetric signatures realizable by matchgates under
 a given basis. Here we state the theorem for a particular
 basis $\begin{bmatrix}1 & 1 \\ 1 & -1 \end{bmatrix}$,
which will be  used in Theorem \ref{thm:csp}.

\begin{theorem}\label{thm:1}
A symmetric signature $[x_0,x_1,\ldots,x_n]$ is
realizable under the basis $\begin{bmatrix}1 & 1 \\ 1 & -1 \end{bmatrix}$ {\em iff} it
          takes one of the following forms:
\begin{itemize}
\item Form 1: there exist  constants $\lambda,s,t$ and
$\epsilon$ where $\epsilon=\pm 1$, such that for all $i, 0\leq i
\leq n$,
\begin{equation*}
x_i=\lambda[(s+t)^{n-i}(s-t)^i+\epsilon
(s-t)^{n-i}(s+t)^i].
\end{equation*}

\item Form 2: there exist a constant $\lambda$, such
that for all $i, 0\leq i \leq n$,
\begin{equation*}
x_i=\lambda[(n-i)(-1)^i  + i
(-1)^{i-1}].
\end{equation*}

\item Form 3: there exist a constant $\lambda$, such
that for all $i, 0\leq i \leq n$,
\begin{equation*}
x_i=\lambda[(n-2)i ].
\end{equation*}
\end{itemize}

\end{theorem}

%% file: interpolation.tex
\section{Polynomial Interpolation}\label{sec:interpolation}
%

In this section, we discuss the interpolation method  we  will
use in this paper.
Polynomial interpolation is a
powerful tool in the study of counting problems initiated by
Valiant~\cite{Valiant79b} and further developed by Vadhan, Dyer and
Greenhill~\cite{Vadhan01,DyerG00} and others. The method we use here is
essentially the same as Vadhan~\cite{Vadhan01}.

For some set of signatures ${\cal F}$, suppose
we want to show that for all
unary signatures $f=[x,y]$, we have ${\rm Holant}({\cal F}\cup
\{[x,y]\}) \leq_{\tt T}  {\rm Holant}({\cal F})$. Let $\Omega = (G, {\cal
F}\cup \{[x,y]\}, \pi)$.  We want to compute ${\rm Holant}_{\Omega}$
in polynomial time  using an oracle for ${\rm Holant}({\cal F})$.

Let $V_f$ be the subset of vertices in $G$ assigned $f$ in $\Omega$.
Suppose $|V_f|=n$.  We can classify all 0-1 assignments $\sigma$ in
the Holant sum according to how many vertices in $V_f$ whose
incident edge is assigned a 0 or a 1.
 Then the Holant value can be expressed as
\begin{equation}\label{holant-as-sum-cijk}
{\rm Holant}_{\Omega} = \sum_{0 \le i \le n}c_{i}x^i y^{n-i},
\end{equation}
where $c_{i}$ is the  sum over all edge assignments $\sigma$, of
products of evaluations at all $v \in V(G) - V_f$, where $\sigma$ is
such that exactly $i$ vertices in $V_f$ have their incident edges
assigned 0 (and $n-i$ have  their incident edges assigned 1.) If we
can evaluate these $c_{i}$, we can evaluate ${\rm Holant}_\Omega$.

Now suppose $\{G_s\}$ is a sequence of ${\cal F}$-gates, and each
$G_s$  has one dangling edge. Denote the signature of $G_s$ by
 $f_s = [x_s,y_s]$, for $s=0, 1, \ldots$. If we
replace each occurrence of $f$ by $f_s$ in $\Omega$ we get  a new
signature grid $\Omega_s$, which is an instance of ${\rm
Holant}({\cal F})$, with
\begin{equation}\label{interpolation-linear-sys}
{\rm Holant}_{\Omega_s} = \sum_{0 \le i \le n}c_{i}x_s^iy_s^{n-i}.
\end{equation}
One can evaluate ${\rm Holant}_{\Omega_s}$ by oracle access to ${\rm
Holant}({\cal F})$. Note that the same set of values $c_{i}$ occurs.
We can treat $c_{i}$ in (\ref{interpolation-linear-sys}) as a set of
unknowns in a linear system.
The idea of interpolation is to find a suitable sequence $\{f_s \}$
such that the  evaluation of ${\rm Holant}_{\Omega_s}$ gives a
linear system (\ref{interpolation-linear-sys}) of full rank, from
which we can solve all $c_{i}$.

In this paper, the sequence  $\{G_s\}$ will be constructed
recursively using suitable gadgetry. There are two gadgets in a
recursive construction: one gadget has arity 1, giving the initial
signature $g=[x_0, y_0]$; the other has arity 2, giving the
recursive iteration. It is more convenient to use a $2\times 2$
matrix $A$ to denote it.
 So we can
recursively connect them as in Figure \ref{Figure:gA} and get
$\{G_s\}$.

\begin{figure}[htbp]
\begin{center}
\includegraphics[width=3 in]{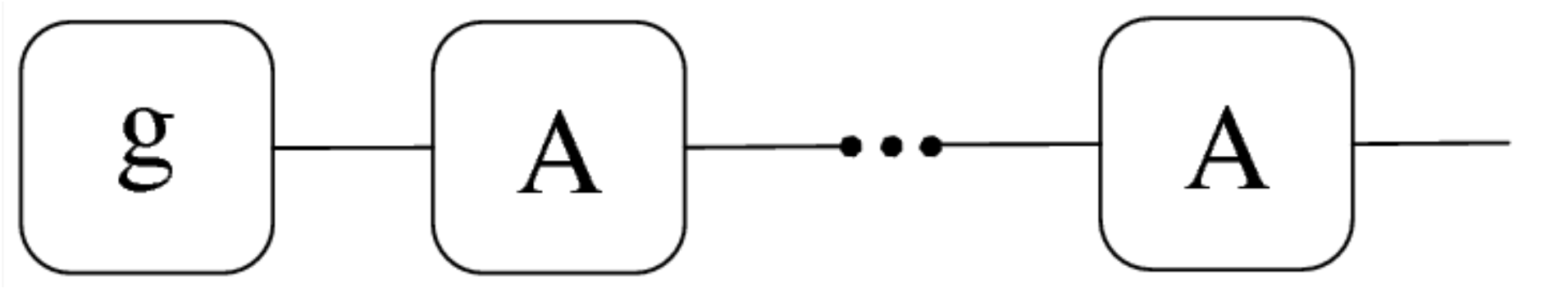}
\caption{Recursive construction.}
\label{Figure:gA}
\end{center}
\end{figure}


The signatures of $\{G_s\}$ have the following relation,
\begin{equation}\label{sig-rec-rel-matrix}
\begin{bmatrix} x_{s}\\ y_{s} \end{bmatrix}=
\begin{bmatrix} a_{11} &  a_{12}    \\
    a_{21}   &  a_{22}   \end{bmatrix}
    \begin{bmatrix} x_{s-1}\\ y_{s-1} \end{bmatrix},
\end{equation}
where $A=\begin{bmatrix} a_{11} &  a_{12}    \\  a_{21}   & a_{22}
\end{bmatrix}$ and $g = \begin{bmatrix} x_0 \\ y_0 \end{bmatrix}$.

We call this gadget pair $(A,g)$ a recursive construction. It
follows from Lemma 6.1 in \cite{Vadhan01} that

\begin{lemma}\label{lemma:interpolation_conditions}
Let $\alpha,\beta$ be the two eigenvalues of $A$. If the following
three conditions are satisfied
\begin{enumerate}
  \item $\det(A)\neq 0$;
  \item $g$ is not a column eigenvector of $A$ (nor the zero
vector);
  \item $\alpha/\beta$ is not a root of unity;
\end{enumerate}
then the recursive construction $(A,g)$ can be
used to interpolate all  unary signatures.
\end{lemma}

A  similar interpolation method also works for signatures with
larger arity but  have two degrees of freedom. For example,
all signatures of form $[0,x,0,y]$. This  is used in
the proof of Lemma \ref{lemma:010x}.

%% file: holant.tex
\section{Dichotomy for Planar Holant$^c$ Problems}\label{section:holant}
Before presenting the
 main dichotomy theorem
for planar Holant$^c$ problems, we prove the
following theorem, which plays a crucial role in the proof of the
main theorem.

\begin{theorem}\label{thm:a010b}
Let $a, b \in  {\mathbb R}$.
\begin{itemize}
\item
If  $ab \not =1$ then Pl-Holant$^c([a,0,1,0,b])$ is \#P-hard.
\item
If  $ab =1$ then Pl-Holant$^c([a,0,1,0,b])$ is solvable in P.
\end{itemize}
\end{theorem}

We first prove three
 lemmas which will be used in the proof of this theorem.

\begin{lemma} \label{thm:a000b-and-010x}
Let $a, b, x \in {\mathbb R}$,
 $ab\neq 0$ and $x\neq \pm 1$. Then Pl-Holant$^c(\{[a,0,0,0,b],[0,1,0,x]\})$ is \#P-hard.
\end{lemma}
\begin{proof}

Firstly, we show how to realize $(=_6)  = [1,0,0,0,0,0,1]$ by $[a,0,0,0,b]$.
 $[a,0,0,0,b]$ can be attached to a vertex of degree 4.
We can connect 3 pairs of edges
 of two copies of $[a,0,0,0,b]$ to realize the   binary function
$[a^2,0,b^2]$.

If $a^2  = b^2$, then we connect one pair of
edges from two copies of
$[a,0,0,0,b]$ to get
$[a^2,0,0,0,0,0,b^2]$.
This is the same as  $(=_6)  = [1,0,0,0,0,0,1]$
after factoring out the non-zero factor $a^2  = b^2$.

If $a^2 \not = b^2$, then  we connect  $[a,0,0,0,b]$ with a
chain of $[a^2,0,b^2]$ of  length $i$ to get
$[a^{2i+1},0,0,0,b^{2i+1}]$. Because for any $i \neq j$,
$a^{2i+1}/b^{2i+1} \neq a^{2j+1}/b^{2j+1}$, we can realize
$(=_4) = [1,0,0,0,1]$ using polynomial interpolation,
as follows.
Consider any signature grid  on a planar
graph $G$ with $n$ occurrences of $=_4$
together with some other signatures. Let $x_{k,\ell}$ be the
sum, over all 0-1 edge assignments $\sigma$,
 of the products of all other vertex
function values in $G$ except at $n$ vertices with $=_4$,
 where $k,  \ell \ge 0$ and $k+ \ell=n$,
{\it and} in $\sigma$ exactly  $k$ occurrences of $=_4$ have
input 0,
and exactly  $\ell$  occurrences of $=_4$ have input 1.
The Holant value is $\sum_{k + \ell = n} x_{k,\ell}$.
Now substitute each occurrence of $=_4$ by $[a^{2i+1},0,0,0,b^{2i+1}]$.
The new signature grid has Holant value $\sum_{k + \ell = n} x_{k,\ell}
(a^{k} b^{\ell})^{2i+1}$.  This gives a Vandermonde system
from which we solve for $x_{k,\ell}$.
Now we have $=_4$.
Then we connect two copies of
$=_4$ on one pair of edges to get   $=_6$.

Take a vertex of degree 6 in a planar graph attached with
$=_6$, where the 6 incident edges are its variables.
We will bundle two adjacent variables to form 3 bundles
of 2 edges each.
Then if the inputs are restricted to $\{(0,0),(1,1)\}$
on each bundle, then the function takes value 1 on
 $((0,0),(0,0),(0,0))$ and
$((1,1),(1,1),(1,1))$, and takes value 0 elsewhere.
Thus if we restrict the domain to
$\{(0,0),(1,1)\}$, it is the  ternary {\sc Equality} function $=_3$.

Let $F=[0,1,0,x]$ and let $H(x_1,x_2,y_1,y_2)=\sum_{z=0,1}
F(x_1,y_1,z)F(x_2,y_2,z)$. This $H$ is realizable
by connecting one pair of edges of two copies of $F$.
(See Figure \ref{fig-H}.)
\begin{figure}[htbp]
\begin{center}
\includegraphics[width=0.5\textwidth]{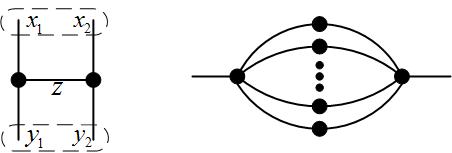}
\caption{ The gadget for function $H$ and $[1,1,x^{2k}]$. }
\label{fig-H}
\end{center}
\end{figure}
We will consider $H$ as a function in
$(x_1,x_2)$ and $(y_1,y_2)$. However we will only
connect $H$ externally by connecting $(x_1,x_2)$
and  $(y_1,y_2)$ to some bundle of two adjacent edges of some  $=_6$.
Since  $=_6$ enforces the values on the bundle to be
either $(0,0)$ or $(1,1)$,
we will only be interested in the restriction of  $H$ to
the domain $\{(0,0),(1,1)\}$.
On this domain, $H$ is a {\it symmetric} function of arity 2,
and can be denoted as $[1,1,x^2]$.
(Note that  $H$ is {\it not} a symmetric function of arity 4 on $\{0,1\}$,
as $H(0,1,0,1) = x$.)


Now we have reduced Pl-Holant$^c(\{[1,0,0,1],[1,1,x^2]\})$ to
Pl-Holant$^c(\{[a,0,0,0,b],[0,1,0,x]\})$.

Using $(=_3) = [1,0,0,1]$, we can realize the {\sc Equality}
function $=_k$ of any arity $k \ge 3$. Then we can realize
$[1,1,x^{2k}]$, for all $k \ge 1$. (See Figure \ref{fig-H}.) If $x
=0$, then we already have $[1,1,0]$.

Suppose $x \not = 0$.
 Because $x^2 \neq
1$ and being a positive real number,
we can realize $[1,1,0]$ by interpolation. Now we
have reduced the problem Pl-Holant$([1,0,0,1] \mid [1,1,0])$
to Pl-Holant$^c(\{[1,0,0,1],[1,1,x^2]\})$.
The bipartite problem
Pl-Holant$([1,0,0,1] \mid [1,1,0])$ is  \#P-hard since it is
counting {\sc Vertex Covers} on planar 3-regular graphs~\cite{XiaZZ07}.
\end{proof}

The following lemma handles a special case of  Theorem~\ref{thm:a010b}.
The proof uses Lemma~\ref{thm:a000b-and-010x}.

\begin{lemma}\label{lemma:00100}
Pl-Holant$^c([0,0,1,0,0])$ is \#P-hard.
\end{lemma}
\begin{proof}

We construct a reduction from  Pl-Holant$^c([1,0,0,0,1],[0,1,0,0])$,
which is \#P-hard by Lemma~\ref{thm:a000b-and-010x},
 to Pl-Holant$^c([0,0,1,0,0])$ by
polynomial interpolation.

Let $F=[0,0,1,0,0]$. There is a series of planar gadgets (a chain of
 $F$) realizing the following sequence of functions:
\[H_2(x_1,x_2,y_1,y_2)=\sum_{x_3,x_4 =0,1}
F(x_1,x_2,x_3,x_4)F(y_1,y_2,x_3,x_4),\]
and for $i \ge 1$,
\[H_{2i+2}(x_1,x_2,y_1,y_2)=\sum_{x_3,x_4 = 0,1}
H_{2i}(x_1,x_2,x_3,x_4)H_2(y_1,y_2,x_3,x_4).\]
 The gadget for $H_{2i}$ is composed of $2i$
functions $F$. As an example, the gadget for $H_4$ is shown in
Figure \ref{Figure:00100}.
\begin{figure}[htbp]
\begin{center}
\includegraphics[width=3 in]{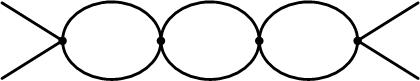}
\caption{ The gadget for $H_4$. } \label{Figure:00100}
\end{center}
\end{figure}

By calculation, $H_{2i}(0,0,0,0)=H_{2i}(1,1,1,1)=1$,
and
$H_{2i}(0,1,0,1)=H_{2i}(0,1,1,0)=H_{2i}(1,0,0,1)=H_{2i}(1,0,1,0)=2^{2i-1}$,
and $H_{2i}$ is zero on other inputs.
Again we will consider the inputs to $H_{2i}$ as bundled into
 $(x_1, x_2)$ and $(y_1, y_2)$.

Given  a planar graph  $G$ as an instance
 of  Pl-Holant$^c([1,0,0,0,1],[0,1,0,0])$,
suppose there are $n$ vertices in $G$ attached with the function
$(=_4) = [1,0,0,0,1]$.
 For $i=1,2,\ldots, n+1$, we construct
an instance $G_i$
of\\
Pl-Holant$^c([0,0,1,0,0])$ as follows:
Replace
each occurrence of $=_4$ by a copy of  $H_{2i}$,
and replace each occurrence of $[0,1,0,0]$ by
$[0,0,1,0,0]$ connected with a $[0,1]$, which exactly realizes
$[0,1,0,0]$.
Note that by replacing  $=_4$ with  $H_{2i}$,
we have bundled  two adjacent edges together (in the planar embedding)
for each vertex attached with  $=_4$.

Let $x_{a,b}$ denote the summation, over all 0-1 edge
assignments $\sigma$, of the products of all other vertex
function values in $G$ except at those
 $n$ vertices with $=_4$,
 where $a, b \ge 0$ and $a+b=n$,
{\it and} in $\sigma$ exactly  $a$ occurrences of $=_4$ have
inputs $\{0000, 1111\}$,
and exactly  $b$  occurrences of $=_4$ have
inputs $\{0101,0110,1001,1010\}$.

Note that
the Holant value on $G_i$ is
\[\sum_{a+b=n} x_{ab} 1^a  (2^{2i-1})^b.\]
On the other hand, the value of
 Pl-Holant$^c([1,0,0,0,1],[0,1,0,0])$ on $G$
is exactly $x_{n,0}$.

When we take $1 \le i \le n+1$,
we get a
 system of linear equations in $x_{ab}$, whose coefficient
matrix is a full ranked  Vandermonde matrix.
Solving this Vandermonde system we obtain the value $x_{n,0}$.

\end{proof}

The following result can be proved by interpolation as well.

\begin{lemma}\label{lemma:interpolate010x}
Let $a \not \in \{-1,0,1\}$ be a real number. Then we can interpolate
all $[x,0,y,0]$ and $[0,y,0,x]$
for $x,y\in \mathbb{C}$ starting from either $[0,1,0,a]$ or  $[a,0,1,0]$.
\end{lemma}
\begin{proof}
The recursive construction is depicted by Figure \ref{Figure:010x}.
\begin{figure}[htbp]
\begin{center}
\includegraphics[width=4 in]{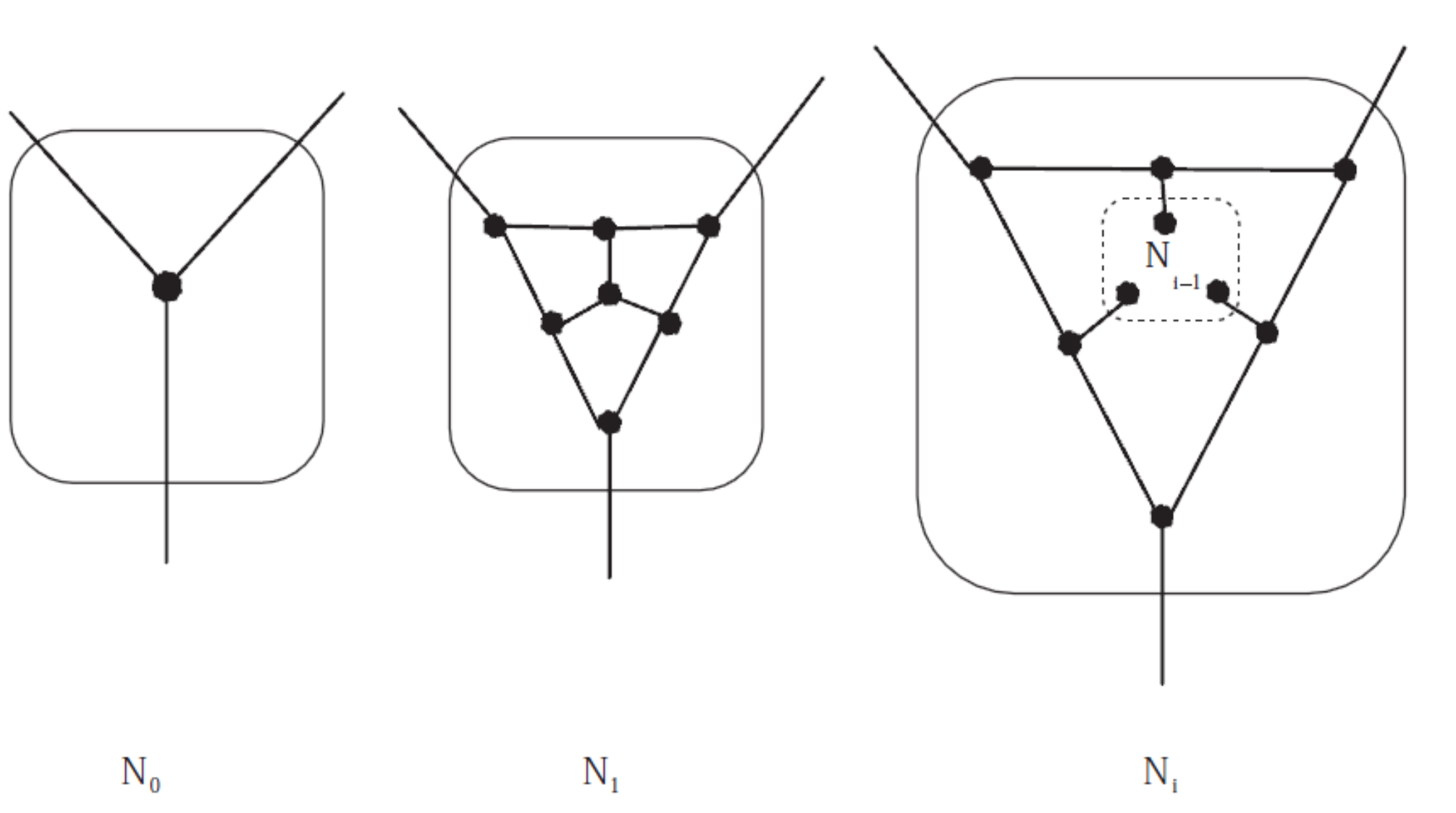}
 \caption{The recursive
construction. The signature of every vertex in the gadget is
$[0,1,0,a]$. } \label{Figure:010x}\end{center}
\end{figure}
By a simple parity argument, every \FF  ~$N_i$ has a signature of
the form $[0,x_i,0,y_i]$. After some calculation, we see  that
they satisfy the following recursive relation:
\[\begin{bmatrix}x_{i+1} \\ y_{i+1} \end{bmatrix}=
\begin{bmatrix}3(a^2+1) & a^3+a \\ 3(a^3+a) & a^6+1 \end{bmatrix}
\begin{bmatrix}x_{i} \\ y_{i} \end{bmatrix}.\]
The signatures we want to interpolate are of arity $3$.
But since all of them take the form $[0,x_i,0,y_i]$ with two
degrees of  freedom, we can use the interpolation method in
Section \ref{sec:interpolation}.
Now we verify that the conditions of that theorem are satisfied.
Let $A=\begin{bmatrix}3(a^2+1) &
(a^3+a) \\ 3(a^3+a) & a^6+1
\end{bmatrix}$, then $(A, [1,a]^{\tt T})$ forms a recursive
construction. Since ${\rm det}(A)=3(a^4-1)^2\neq 0$, the first
condition holds. Its characteristic equation is $X^2-(a^6+3 a^2+4)X
+3(a^4-1)^2=0$. For this quadratic equation, the discriminant
 $\Delta=(a^6-3 a^2
-2)^2+ 12 (a+a^3)^2 >0$. So $A$ has two distinct real eigenvalues.
The sum of the two eigenvalues is ${\rm tr} A
= a^6+3 a^2+4 > 0$.
So they are not opposite to each other. Therefore, the ratio
of these two eigenvalues is not a root of unity and the third
condition holds. Consider the second condition: if the initial
vector $[1,a]^{\tt T}$ is a column eigenvector of $A$, then we have
$A\begin{bmatrix}1
\\ a
\end{bmatrix}= \lambda \begin{bmatrix}1 \\ a
\end{bmatrix} $, where $\lambda$ is an eigenvalue of $A$. From
this, we will conclude that $a(a^2-1)(a^4-1)=0$,
which can not happen given $a \not \in \{-1,0,1\}$. To sum up, this
recursive relation satisfies all three conditions of Lemma
\ref{lemma:interpolation_conditions} and can be used to interpolate
all  signatures of the form $[0,x,0,y]$. This completes the proof.
\end{proof}

\noindent {\bf Proof of Theorem \ref{thm:a010b}}
If $ab=1$, then $[a,0,1,0,b]$ is realizable by some matchgate,
by Lemma~\ref{even-geometric-series-lemma}.
This realizability
also applies to the unary functions $[1,0]$ and $[0,1]$.
Hence the  problem Pl-Holant$^c([a,0,1,0,b])$
 can be solved in polynomial time by matchgate computation via
the FKT method~\cite{TF1961,Kasteleyn1961,Kasteleyn1967}.
In the following we assume that $ab\neq 1$ and prove that
the problem is \#P-hard.
The case $a=b=0$ is proved in Lemma \ref{lemma:00100}. Now we can assume at least one of $a$ and $b$ is non-zero,
and by symmetry we assume $a\neq 0$.

We know from our dichotomy for Holant$^c$ problems \cite{STOC09}
that Holant$^c([a,0,1,0,b])$ for general graphs
is \#P-hard unless $a=b=1$ or $a=b=-1$, in which cases it is tractable.
Both of these tractable  cases are also included in the
tractable cases ($ab=1$) here.
Therefore, if we can realize a
 {\it cross function} $X$ with a planar gadget when $ab\neq 1$,
 we can reduce  Holant$^c([a,0,1,0,b])$
for general graphs to   Pl-Holant$^c([a,0,1,0,b])$ and finish the proof.
Here a cross function $X$ has 4 input bits,
and satisfies  $X_{0000}= X_{0101}=X_{1010} = X_{1111}=1$
and $X_{\alpha}=0$ for all other inputs $\alpha \in \{0, 1\}^4$.

If $\{a,b\}\not \subset \{-1,0,1\}$,
we can use Lemma \ref{lemma:interpolate010x}
to interpolate all  $[x,0,y,0]$, for $x, y \in {\mathbb C}$.
If $\{a,b\}  \subset \{-1,0,1\}$, then
there are only four cases: $[1,0,1,0,-1]$, $[1,0,1,0,0]$, $[-1,0,1,0,1]$
and $[-1,0,1,0,0]$. In all four cases, it is easy to
verify that we can realize a signature with a
form $[c_1,0,c_2,0]$ where $c_1 c_2 \neq 0$ and $c_1 \neq \pm c_2$ using the
gadget in Figure \ref{Figure:1010-1}. ( For $[1,0,1,0,-1]$, we get $[8,0,4,0]$ by using $[1,0]$ in the gadget;
 for $[1,0,1,0,0]$, we get $[8,0,5,0]$ by using $[1,0]$;  for $[-1,0,1,0,1]$, we get $[0,4,0,8]$ by using $[0,1]$;
 and for  $[-1,0,1,0,0]$, we get $[0,1,0,3]$ by using $[0,1]$.)
After factoring out a nonzero factor, we have $[c', 0,1,0]$,
where $c' \in {\mathbb R}$ and  $c' \not \in \{0, \pm 1\}$.
As a result, we can also interpolate all  $[x,0,y,0]$,
where $x, y \in {\mathbb C}$.

\begin{figure}[htbp]
\begin{center}
\includegraphics[width=3 in]{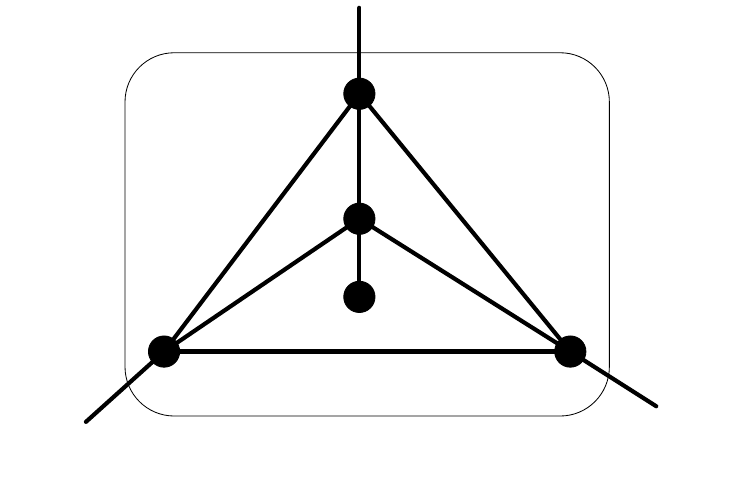}
\caption{ The
signature of the degree $1$ vertex in the gadget is  $[1,0]$ or $[0,1]$. }
\label{Figure:1010-1}
\end{center}
\end{figure}

Now we can use all signatures of
the form $[x,0,y,0]$,
for arbitrary $x, y \in {\mathbb C}$,
 to build new gadgets. We also have all $[x,0,y]$
by connecting  $[x,0,y,0]$  to a $[1,0]$.
By connecting a $[\sqrt[4]{t/a},0,\sqrt[4]{a/t}]$
to each edge of the signature $[a,0,1,0,b]$,
we get  $[t,0,1,0,\frac{c}{t}]$ for all $t\neq 0$, where $c=ab\neq 1$.
Using all these, we will build a
 planar gadget in Figure \ref{fig:cross} to realize the cross function
$X$.
In the equations below
 $x, y, t$ are three variables we can set to any complex numbers,
with $t \not = 0$. The parameter $c$ is given and not equal to 1.

(Of course we presumably could not build a cross function $X$ if $c=1$;
this is {\it exactly} when the problem is in P, and this
is also  {\it exactly} when our construction of $X$ fails.
If a  cross function $X$ were to exist when $c=1$ then
P = \#P would follow.  However, it is still rather mysterious
that algebraically $c=1$ is  {\it exactly} when our construction
fails. This failure condition is by no means obvious from
the equations below.)

\begin{figure}[htbp]
\begin{center}
    \includegraphics[width=3 in]{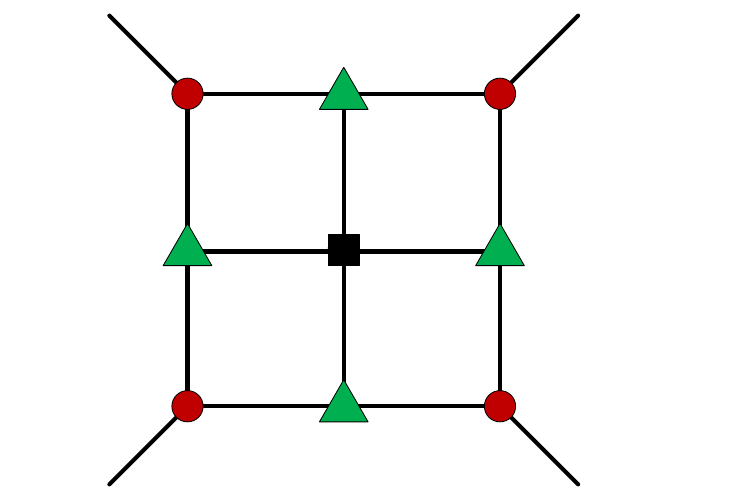}
    \caption{This gadget is to realize the Cross function. The signature for the center vertex (black and square) is $[t, 0,1,0, \frac{c}{t}]$.
    The signature for the vertexes in the four corners (red and circle) is $[x,0,1,0]$. The signature for the vertexes in the middle
    of the boundaries (green and triangle)  is $[y,0,1,0]$. }
    \label{fig:cross}
\end{center}
\end{figure}

We can compute the signature of the gadget
in Fig.~\ref{fig:cross}. If the input has
an odd number of $1$s, the value is $0$. For other inputs, we have
\begin{eqnarray*}
X_{0000} &=& x^4 y^4 t+t+4 x^3 y^2+4 x+4 x^2 y+\frac{2 c x^2}{t}\\
X_{1111} &=& 2 y^2 t+12 y+\frac{2c}{t} \\
X_{0101} = X_{1010} &=& 2 x y^2 t+4 x^2 y^2+4+4 x y+\frac{2 c x}{t}\\
X_{0011}=X_{1001}=X_{1100}=X_{0110} &=&x^2 y^3 t+y t+3 x^2 y^2+3+6 x y+\frac{2 c x}{t}.
\end{eqnarray*}

Here we prove that for any $c \not = 1$, we can assign
suitable complex values to $x,y$ and $t$, where $t \not =0$,
such that $A = B = C \not = 0$, and $D=0$,
where $A$, $B$, $C$ and $D$ denote respectively the
four functions of $x, y$ and $t$ listed in the four lines above.
%

\begin{myclaim}
For any $c \not = 1$,
\[(x-1)^2 = \frac{16}{c-1}\]
has a solution $x  \not \in \{0, +1, -1\} $.
This $x$ satisfies
\begin{equation}\label{lm-1-eqn1}
\left( 2 - \frac{x(x+3)}{x-1} \right)
\left(  \frac{x+3}{x-1} \right)
+ cx + 6 =0 \textrm{.}
\end{equation}
\end{myclaim}

\begin{proof}
Clearly $x=1$ is not a solution to $(x-1)^2 = \frac{16}{c-1}$.
Also the equation  has two distinct roots.
When $c=17$ there is a solution $x=2 \not \in \{0, +1, -1\} $.
When $c \not = 17$, we can verify $x=0$ is not a solution.
Hence the equation always has a solution other than $0, \pm 1$.

To verify (\ref{lm-1-eqn1})
we have
\begin{eqnarray*}
& & (2x -2 - x^2 -3x)(x+3) + (cx +6)(x^2 -2x +1)\\
&=& -(x^3 + 4x^2 + 5x + 6) + c x^3 + (6-2c) x^2 +(-12 +c) x + 6 \\
&=& (c-1) x^3 - 2 (c-1) x^2 + (c-17) x\\
&=& (c-1) x [ (x-1)^2 -{16}/(c-1) ]\\
&=& 0.
\end{eqnarray*}
\end{proof}

Now we fix $x  \not \in \{0, +1, -1\} $ satisfying (\ref{lm-1-eqn1})
for any given $c \not = 1$.

\begin{myclaim}
For any $c \not = 1$, we can pick $z \not = \pm 1$ such that
\begin{equation}\label{lm-2-eqn1}
\frac{4 z}{(1+z)^2} = \frac{ x(x+3)}{x-1}.
\end{equation}
\end{myclaim}

\begin{proof} We are given  $x \not = 0, \pm 1$.
If $x=-3$, we can pick $z =0$.
Now suppose $x \not =-3$.  Consider the quadratic
equation in $z$
\[4 z (x-1) = x (x+3) (1+z)^2.\]
This is quadratic since $x (x+3) \not =0$.
We can check that $z= + 1$ (and $-1$ respectively) is not a solution,
as this would force $x = - 1$ (and $+1$ respectively).
However, any solution where  $z \not = -1$ and $x \not = 1$
is equivalent to (\ref{lm-2-eqn1}).
Hence we have a solution $z \not = \pm 1$ to (\ref{lm-2-eqn1}).

\end{proof}

Now we further fix a  $z \not = \pm 1$
satisfying (\ref{lm-2-eqn1}), and let $y = z/x$ such that $xy \not = \pm 1$,
for any $c \not = 1$.

\begin{myclaim}
For any  $c \not = 1$, there exist $x \not \in \{0, +1, -1\} $
and $y$ such that $xy \not = \pm 1$ satisfying
\begin{equation}\label{lm-3-eqn2}
\frac{2 (1 + x^2 y^2)}{(1+ xy)^2} \cdot \frac{x+3}{x-1}
+ cx + 6 = 0.
\end{equation}
\end{myclaim}

\begin{proof}
%
\begin{eqnarray*}
& & \frac{2 (1 + x^2 y^2)}{(1+ xy)^2} \cdot \frac{x+3}{x-1}
+ cx + 6 \\
& = & 2 \left( 1 - \frac{2z}{(1+ z)^2} \right) \cdot \frac{x+3}{x-1}
+ cx + 6 \\
& = & \left( 2 - \frac{x (x+3)}{x-1} \right) \cdot \frac{x+3}{x-1}
+ cx + 6 \\
& = & 0.
\end{eqnarray*}
Here we used (\ref{lm-2-eqn1}) and (\ref{lm-1-eqn1}).

\end{proof}

Now we will set $t = 4/(1+xy)^2$.  Clearly $t \not =0$.
We next verify that $D = 0$.

By (\ref{lm-2-eqn1}) and (\ref{lm-3-eqn2}) we get
\[\frac{8 y (1 + x^2 y^2)}{(1+ xy)^4} + cx + 6 = 0.\]
Then
\[t^2 y (1 + x^2 y^2) + 2 c x +  3t (1+ xy)^2 =0.\]
Thus
\[D = yt (1 + x^2 y^2) + 3  (1+ xy)^2 + \frac{2 c x}{t} = 0.\]

Next we show that $C = \frac{4 (1-xy)^2}{1-x} \not =0$.

By $D=0$, we have
\[
C = 2 x y^2 \frac{4}{(1+ xy)^2} + 4 (1+ xy)^2 - 4 xy +
[- yt (1 + x^2 y^2) - 3 (1+ xy)^2].
\]
Hence
\begin{eqnarray*}
C &=& \frac{8xy^2}{(1+ xy)^2} + (1+ xy)^2 - 4 xy
- y \frac{4 (1 + x^2 y^2)}{(1+ xy)^2}\\
 &=& \frac{4y}{(1+ xy)^2} \left[ 2xy - 1 - x^2 y^2 \right]
+ (1-xy)^2 \\
 &=& \left( \frac{-4y}{(1+ xy)^2} + 1 \right) (1-xy)^2 \\
 &=& \frac{4 (1-xy)^2}{1-x} \not = 0,
\end{eqnarray*}
using (\ref{lm-2-eqn1}).

The next task is to show $B=C$.

We have
\[ C = 4 (1-xy)^2 + x B.\]
Hence
\[ B= \frac{1}{x}
\left[  \frac{4 (1-xy)^2}{1-x} -  4 (1-xy)^2 \right]
= \frac{4 (1-xy)^2}{x} \left[ \frac{1}{1-x} -1 \right]
= \frac{4 (1-xy)^2}{1-x}
=C.\]

Finally we verify $A=C$ as well.

\[A = (x^4 y^4 +1) t + x [ C - 2 x y^2 t ]
= C + (x-1)C - 2 x^2 y^2 t + (x^4 y^4 +1) t
= C - 4 (1-xy)^2 + t (x^2 y^2-1)^2
= C.\]

\qed

Now we come to the main dichotomy theorem for Pl-Holant$^c$
problems.

\begin{theorem}\label{thm:holant-c}
Let $\mathcal{F}$ be a set of real symmetric signatures.
{\rm Pl-Holant}$^c(\mathcal{F})$ is \#P-hard unless $\mathcal{F}$
satisfies one of the following conditions, in which case it is
tractable:
\begin{enumerate}
    \item ${\rm Holant}^c(\mathcal{F})$ is tractable (for which
we have an effective dichotomy~\cite{STOC09}); or
    \item Every signature in $\mathcal{F}$ is realizable by some matchgate
(for which we have a complete characterization \cite{MGI}).
\end{enumerate}
\end{theorem}

Before we give the proof, we do some normalization of the signature set
$\mathcal{F}$.
Since any degenerate
signature $[x,y]^{\otimes k}$ can be replaced by the corresponding
unary signature $[x,y]$ without changing the complexity of the
problem, we always assume that all the signatures in ${\cal F}$,
whose arity is greater than $1$, are non-degenerate. Since $[1,0]$ and
$[0,1]$ are freely available, we can construct any sub-signature of
an original signatures as well as any signature realizable by some
\FF .

The main idea of the  proof is to interpolate
all  unary functions. If we can do that, we can reduce the problem
 {\rm Pl-Holant}$^* ({\cal
F})$ to {\rm Pl-Holant}$^c ({\cal F})$ and finish the proof.
We note that our dichotomy in \cite{STOC09} for  ${\rm Holant}^* ({\cal
F})$ also holds for planar graphs. In some cases, we cannot
interpolate all unary functions, then we prove the theorem separately,
mainly using Lemma~\ref{thm:a000b-and-010x}
and Theorem~\ref{thm:a010b}. The following lemma is for interpolation
of unary functions.

\begin{lemma}\label{lemma:interpolation}
If we can construct from ${\cal F}$  a gadget  with signature $[a,b,c]$, where
$b^2\neq a c$, $b\neq 0$ and $a+c \neq 0$, then we can interpolate
all  unary functions. (Hence the conclusions
of
Theorem \ref{thm:holant-c}
hold.)
\end{lemma}

\begin{proof}
we use the interpolation method as described in Section
\ref{sec:interpolation}. We consider two recursive constructions
$(\begin{bmatrix} a & b \\ b & c \end{bmatrix},
\begin{bmatrix}1 \\ 0 \end{bmatrix} )$ and $(\begin{bmatrix} a & b \\ b & c \end{bmatrix},
\begin{bmatrix}0 \\ 1 \end{bmatrix} )$, and argue that at least one
of them will succeed given the conditions on $a,b,c$.
We use $A$ to denote $\begin{bmatrix} a & b \\ b & c
\end{bmatrix}$. Since $b^2 \neq ac$, $A$ is non-degenerate,
the first condition of Lemma \ref{lemma:interpolation_conditions} is
satisfied for both recursive constructions. If both
$[1,0]$ and $[0,1]$ are column eigenvectors of $A$, then $b=0$, a
 contradiction. So at least for one of the two recursive constructions,
 the second condition of Lemma \ref{lemma:interpolation_conditions} is satisfied.
 Since $A$ is a real symmetric matrix,  both its eigenvalues are real.
 If the ratio of two real numbers is a root of unity,
 they must be the same or opposite to each other.
 If the two  eigenvalues are the same, we have $b=0$ and $a=c$,
 a contradiction. If the two eigenvalues are opposite to each other,
 then we have $a+c=0$, also a contradiction.
Therefore, the third condition of Lemma
\ref{lemma:interpolation_conditions} is also satisfied for both
recursive constructions. To sum up, at least one of the two
recursive constructions satisfies all the conditions of Lemma
\ref{lemma:interpolation_conditions}.
  This completes the proof.
\end{proof}



If we can construct from ${\cal F}$
 a gadget with a binary symmetric signature
$[a,b,c]$, which satisfies all the conditions in Lemma
\ref{lemma:interpolation}, then we are done. For most cases, we
prove the theorem by interpolating all  unary signatures.
However, in some more delicate
 cases, we are not able to do that. For example, if
all  signatures from ${\cal F}$
 have the parity condition,
which includes a proper superset of matchgate signatures, then all  unary
signatures we can realize have form $[a,0]$ or $[0,a]$, so we can
not interpolate all  unary signatures. For these cases, our starting
point is  Theorem \ref{thm:a010b}.

We define some families of symmetric signatures, which will be used
in our proof.
\begin{eqnarray*}
{\cal G}_1&=&\{[a,0,0,\cdots,0,b] ~\mid~ab \neq 0 \}\\
{\cal G}_2&=&\{[x_0,x_1,\cdots,x_k] ~\mid~ \forall i \mbox{ is even}, x_i=0 \mbox{ or } \forall i \mbox{ is odd}, x_i=0  \}\\
{\cal G}_3&=&\{[x_0,x_1,\cdots,x_k] ~\mid~ \forall i, x_i+x_{i+2}=0  \}\\
{\cal M} &=&\{~f ~\mid~ f \mbox{ is realizable by some matchgate } \}.
\end{eqnarray*}
We note that ${\cal G}_1$, ${\cal G}_2$ and ${\cal G}_3$ are
supersets of ${\cal F}_1$, ${\cal F}_2$ and ${\cal F}_3$
respectively. Furthermore (the real part of)
 ${\cal F}_2 \subseteq {\cal M} \subseteq {\cal G}_2$.
The conditions in ${\cal G}_2$ are called
parity conditions.  The following several lemmas all have the form ``If
${\cal F} \not \subseteq {\cal A}$, then the conclusions
of Theorem \ref{thm:holant-c}
hold." After proving each lemma, in subsequent  lemmas, we only need to
consider the case that ${\cal F} \subseteq {\cal A}$.


\begin{lemma}\label{lemma:inG1G2G3}
If ${\cal F} \not \subseteq {\cal G}_1 \cup {\cal G}_2\cup {\cal
G}_3$, then the conclusions
of
Theorem \ref{thm:holant-c}
hold.
\end{lemma}

\begin{proof}
Since ${\cal F} \not \subseteq {\cal G}_1 \cup {\cal G}_2\cup {\cal
G}_3$, there exists an $f \in {\cal F}$ and  $f \not \in {\cal G}_1
\cup {\cal G}_2\cup {\cal G}_3$. Since all  unary signatures are
in ${\cal G}_3$, the arity of $f$ is greater than $1$ and $f$ is
non-degenerate. There are two cases according to whether $f$ has a
zero entry or not.

(1) $f$ has some zero entries. If there exists a sub-signature of
$f$ of the form $[0,a,b]$ or $[a,b,0]$, where $ab\neq 0$, then we
are done by Lemma \ref{lemma:interpolation}. Otherwise, we can
conclude that there are no two successive non-zero entries. So the
signature $f$ has this form $[0^{i_0}x_1 0^{i_1} x_2 0^{i_2}\cdots
x_k 0^{i_k}] $, where $k \ge 1$,
$x_j \neq 0$ and for all $1\leq j \leq k-1$,
$i_j\geq 1$. If for all $1\leq j \leq k-1$, $i_j $ is odd, (including $k=1$),
 then
$f\in {\cal G}_2$, a contradiction. Otherwise there exists a sub-signature
of form $[x,0,0,\cdots, 0 ,y]$, where $xy\neq 0$ and there
are a positive even number of $0$s between $x$ and $y$. If this is the
entire
$f$, then $f\in {\cal G}_1$, a contradiction. So there is one $0$
before $x$ or after $y$. By symmetry, we assume there is a $0$
before $x$, so we have a sub-signature $[0,x,0,0,\cdots, 0 ,y]$,
whose arity is even and at least 4. We label its dangling edges
$1,2, \cdots, 2k$. Then for every $i=1,2,\cdots, k-1$, we connect
dangling edges $2i+1$ and $2i+2$ together to form a regular edge. After
that, we have an \FF ~with arity 2, and its signature is $[0,x,y]$.
Then we are done by Lemma \ref{lemma:interpolation}.

(2) $f$ has no zero entry. We only need to prove that we can
construct a function $[a',b',c']$ satisfying the three conditions in
Lemma \ref{lemma:interpolation}. Suppose all sub-signatures of $f$
with arity 2 do not satisfy all the three conditions. For each
sub-signature $[a',b',c']$, either $a'+c'=0$, or $b'^2=a'c'$. If all
of them satisfy $a'+c'=0$, then $f \in \mathcal{G}_3$. A
contradiction. If all of them satisfy  $b'^2=a'c'$, then $f$ is
degenerate. A contradiction. W.l.o.g., we can assume there is a
sub-signature $[a,b,c,d]$ of $f$, such that $a+c=0$, $b+d \neq 0$,
and $c^2=bd$. We get this sub-signature $[a,b,c,d]$
by $[1,0]$ and $[0,1]$.  Combining two $[a,b,c,d]$, we can get a function
$[a',b',c']=[a^2+2b^2+c^2,ab+2bc+cd,b^2+2c^2+d^2]
=[2(b^2 + c^2), c (b+d), (b+d)^2]$.
Then $b'=c(b+d) \neq
0$. $a'+c'>0$.
And
$a'c'-b'^2= (b+d)^2 (2b^2 + c^2) >0$. We are
done by Lemma \ref{lemma:interpolation}.
\end{proof}

The following lemma uses Theorem~\ref{thm:a010b}
in an essential way, which in turns depends on the crossover.

\begin{lemma}\label{lemma:inG1MG3}
If ${\cal F} \not \subseteq {\cal G}_1 \cup {\cal M} \cup {\cal
G}_3$,  then the conclusions
of
Theorem \ref{thm:holant-c}
hold.
\end{lemma}

\begin{proof}
If ${\cal F} \not \subseteq {\cal G}_1 \cup {\cal G}_2\cup {\cal
G}_3$, then by Lemma \ref{lemma:inG1G2G3}, we are done. Otherwise,
there exists a signature $f \in {\cal F} \subseteq {\cal G}_1 \cup
{\cal G}_2\cup {\cal G}_3 $ and $f \not \in {\cal G}_1 \cup {\cal
M}\cup {\cal G}_3$. Then it must be the case that $f \in {\cal
G}_2$. Note that every signature with arity at most $3$ in ${\cal
G}_2$ (this is called the parity condition)
is also contained in ${\cal M}$,
so $f$ is
of arity greater  than $3$.  Let $f=[x_0, x_1, \cdots, x_n]$, for some $n\geq 4$.
Suppose
 there exists some $i\in [2,3,\cdots, n-2]$ such that $x_i\neq 0$.
If $x_{i-2} x_{i+2} \neq x_i^2$, then we can get $[x_{i-2}, 0, x_i, 0, x_{i+2}]$
by $[1,0]$ and $[0,1]$ which restrict
the signature to a sub-signature.
Then  the problem is \#P-hard
by Theorem \ref{thm:a010b} and we are done. Otherwise, we  have
$x_{i-2} x_{i+2} = x_i^2 \neq 0$. Then starting from $x_{i-2}\neq 0$ and if
$i-2\in [2,3,\cdots, n-2]$, we can get $x_{i-4} x_{i} = x_{i-2}^2 \neq 0$.
Similarly we can start with $x_{i+2}$.
A signature satisfying the parity condition and is a geometric
series on the alternate entries is realizable by a matchgate~\cite{CaiC06,MGI},
a contradiction.

Now we may assume  $x_i=0$ for all $i\in [2,3,\cdots, n-2]$.
Since $f\in {\cal G}_2 -({\cal M} \cup {G_1}) $, we know that
there are only three possible subcases:
(1) $n$ is odd, $n \ge 5$,
 $x_0 x_{n-1} \neq 0$ and $x_1=x_{n}=0$;
(2) $n$ is odd,   $n \ge 5$,
$x_1 x_n \neq 0$ and $x_0=x_{n-1}=0$;
(3) $n \geq 6$ is even, $x_1 x_{n-1} \neq 0$
and $x_0=x_n=0$.
This uses the theory of matchgate realizability~\cite{CaiC06,MGI}.
Crucially, if $n$ is even and $n <6$, then $n=4$
and the case $x_1 x_{n-1} \neq 0$, $x_0=x_n=0$
belongs to ${\cal M}$.
The subcases (1) and (2) are reversals
 of each other and (3) contains a signature in form (1) and (2).
So after normalizing (and  connecting pairs
of edges together if $n> 5$),  we will get a signature
$[0,1,0,0,0,x]$ where $x\neq 0$. So we have both sub-signature $[0,1,0,0]$ and $[1,0,0,0,x]$.  As we proved in Lemma \ref{thm:a000b-and-010x}, the problem is
\#P-hard and we are done. This finishes the proof.
\end{proof}

\begin{lemma}\label{lemma:010x}
If $[0,1,0,x] \in {\cal F}$ (or $[1,0,x,0] \in {\cal F}$)
where  $x \in {\mathbb R}$, $x\neq \pm 1$, then the conclusions
of
Theorem \ref{thm:holant-c}
hold.
\end{lemma}

\begin{proof}
If $x\neq 0$, we can use Lemma \ref{lemma:interpolate010x} to interpolate $[0,1,0,0]$.
So we assume we have $[0,1,0,0]$ from ${\cal F}$.
If ${\cal F} \not \subseteq {\cal G}_1 \cup {\cal M}\cup {\cal G}_3$,  then by Lemma \ref{lemma:inG1MG3}, we are done.
If ${\cal F} \subseteq
 {\cal M}$, then the problem is tractable and we are done. Otherwise, there exists a signature $f \in {\cal F} \subseteq {\cal G}_1 \cup
{\cal M}\cup {\cal G}_3 $ and $f \not \in {\cal M}$. That is $f \in ({\cal G}_1 \cup {\cal G}_3 -{\cal M})$.

If $f$ has  arity $\ge 1$ and of
the form $[x_0, x_1, -x_0, -x_1, x_0 \cdots]  \in {\cal G}_3 - {\cal M}$, then we will have $x_0 x_1 \neq 0$. Otherwise we would have $f \in {\cal M}$, a contradiction.
Connecting one unary signature $[x_0,x_1]$ to $[0,1,0,0]$, we get $[x_1,x_0,0]$ which satisfies all the conditions in Lemma \ref{lemma:interpolation}, and we are done.

Now we consider $f=[1,0,0,\cdots,0,y] \in {\cal G}_1 - {\cal M}$, where $y\neq 0$. Since $f \not \in {\cal M} $, its arity $n$ is greater than $2$.
If $n$ is odd, we can connect its edges except one to get a unary signature $[1,y]$. Then we can use a similar argument as above and we are done.
If $n$ is even, then it is at least $4$,
since $f \not \in {\cal M}$. After connecting its edges except four, we can get $[1,0,0,0,y]$. Together with $[0,1,0,0]$, we know the problem is
\#P-hard by Lemma \ref{thm:a000b-and-010x}. This completes the proof.
 \end{proof}


\begin{lemma}\label{lemma:inG1F2G3}
If ${\cal F} \not \subseteq {\cal G}_1 \cup {\cal F}_2\cup {\cal
G}_3$, then the conclusions
of
Theorem \ref{thm:holant-c} hold.
\end{lemma}

\begin{proof}
If ${\cal F} \not \subseteq {\cal G}_1 \cup {\cal G}_2\cup {\cal
G}_3$, then by Lemma \ref{lemma:inG1G2G3}, we are done. Otherwise,
there exists a signature $f \in {\cal F} \subseteq {\cal G}_1 \cup
{\cal G}_2\cup {\cal G}_3 $ and $f \not \in {\cal G}_1 \cup {\cal
F}_2\cup {\cal G}_3$. Then it must be the case that $f \in {\cal
G}_2$. Note that every signature with arity less than $3$ in ${\cal
G}_2$ is also contained in ${\cal G}_1 \cup  {\cal G}_3$, so $f$ is
of arity greater than $2$. Since $f \not \in {\cal G}_1$, there is
some non-zero in the middle of the signature $f$, after
normalization, we can assume there is a sub-signature of form
$[0,1,0,x]$ (or $[x,0,1,0]$). If $x\neq \pm 1$, then by
Lemma~\ref{lemma:010x}, we are done.
Otherwise, for every such pattern, we have $x=\pm 1$. Since $f \not
\in {\cal F}_2$, then there is some sub-signature $[0,1,0,-1]$ and
because $f \not \in {\cal G}_3$, there is some sub-signature
$[0,1,0,1]$. Therefore, there is a sub-signature $[1,0,1,0,-1]$ of
$f$. Then by Theorem~\ref{thm:a010b}, we
know that the problem is \#P-hard and we are done. This completes
the proof.
\end{proof}

\begin{lemma}\label{lemma:inG1G3}
If ${\cal F} \not \subseteq {\cal G}_1 \cup  {\cal G}_3$, then the conclusions
of
Theorem \ref{thm:holant-c} hold.
\end{lemma}

\begin{proof}
If ${\cal F} \not \subseteq {\cal G}_1 \cup {\cal F}_2\cup {\cal
G}_3$, then by Lemma \ref{lemma:inG1F2G3}, we are done. Otherwise,
there exists a signature $f \in {\cal F} \subseteq {\cal G}_1 \cup
{\cal F}_2\cup {\cal G}_3 $ and $f \not \in {\cal G}_1 \cup {\cal
G}_3$. Then it must be the case that $f \in {\cal F}_2$. Note that
every signature with arity less than $3$ in ${\cal F}_2$ is also
contained in ${\cal G}_1 \cup  {\cal G}_3$, so $f$ is of arity
at least 3. Then $f$ has a sub-signature $[1,0,1,0]$ or
$[0,1,0,1]$. By symmetry, we assume it is $[1,0,1,0]$. If ${\cal F}
\subseteq {\cal F}_1 \cup {\cal F}_2 \cup {\cal F}_3 $, then Theorem
\ref{thm:holant-c} trivially holds and there is nothing to prove. If
not, there exists a signature $g \in {\cal F}
- {\cal F}_1 \cup {\cal F}_2 \cup {\cal F}_3 $.
By ${\cal F}  \subseteq {\cal G}_1 \cup {\cal F}_2\cup {\cal
G}_3$,
 either $g \in {\cal G}_1-{\cal F}_1 \cup
{\cal F}_2 \cup {\cal F}_3 ~(\subseteq{\cal G}_1-{\cal F}_1)$ or
$g \in {\cal G}_3-{\cal F}_1 \cup {\cal F}_2 \cup {\cal F}_3
~(\subseteq{\cal G}_3-{\cal F}_3)$.

For the first case,  $g \in ({\cal G}_1-{\cal F}_1)$, after a
scale, $g$ is of form $[1,0,0,\cdots, b]$, where $b \not \in
\{-1,0,1\}$. If the arity of $g$ is odd, we can realize $[1,b]$. (We
connect  every two adjacent dangling edges into one edge and leave one
dangling edge.) Then connecting this unary signature to one dangling
edge of $[1,0,1,0]$,
 we can realize a binary
signature $[1,b,1]$.
Then by Lemma \ref{lemma:interpolation}, Theorem \ref{thm:holant-c}
holds. If the arity of $g$ is even, we can realize $[1,0,b]$ (leave
two dangling edges). By connecting  one of its dangling edge to one
dangling edge of $[1,0,1,0]$, we can have a new ternary signature
$[1,0,b,0]$.
By Lemma \ref{lemma:010x}, we are done.

For the second case $g \in ({\cal G}_3-{\cal F}_3)$, $g$ has a
sub-signature of form $[1,b]$, where $b \not \in \{-1,0,1\}$. By the
same argument as above, Theorem \ref{thm:holant-c} holds. This
completes the proof.
\end{proof}

\begin{lemma}\label{lemma:inG1F3}
If ${\cal F} \not \subseteq {\cal G}_1 \cup  {\cal F}_3$, then the conclusions
of
Theorem \ref{thm:holant-c} hold.
\end{lemma}

\begin{proof}
If ${\cal F} \not \subseteq {\cal G}_1 \cup {\cal G}_3$, then by
Lemma \ref{lemma:inG1G3}, we are done. Otherwise, there exists a
signature $f \in {\cal F} \subseteq {\cal G}_1 \cup {\cal G}_3 $ and
$f \not \in {\cal G}_1 \cup {\cal F}_3$. Then it must be the case
that $f \in {\cal G}_3$, and $f$ has a sub-signature of form
$[1,a,-1]$, where $a \not \in \{-1,0,1\}$.

If ${\cal F}  \subseteq \{[1,0,1] \} \cup  {\cal G}_3$, then ${\rm
Holant}^* ({\cal F})$  is polynomial time computable by Theorem
\ref{thm-holant-star} and as a result Theorem \ref{thm:holant-c} trivially
holds and we are done.
%

If not, there exists a signature $g \in {\cal F}  \subseteq {\cal
G}_1 \cup  {\cal G}_3$ and $g \not \in \{[1,0,1] \} \cup  {\cal
G}_3$. Then it must be the case that $g \in {\cal G}_1$. The arity
of $g$ is greater than $1$,
as $g \not \in {\cal G}_3$.


If the arity of $g$ is $2$, then $g$ is of form $[1,0,b]$, where $b
\not \in \{-1,0,1\}$. Connecting two signatures $[1,0,b]$ to
both sides of one binary signature $[1,a,-1]$, we can get a new
binary signature $[1,ab, -b^2]$. It satisfies all the conditions of
Lemma \ref{lemma:interpolation}, and we are done. If the arity of
$g$ is greater than $2$, then we can always realize a signature
$[1,0,0,b]$, where $b\neq 0$. (We connect the unary signature
$[1,a]$ to all its dangling edges except the three ones.)
 Then we can use an \FF ~in Figure \ref{Figure:[1,a^2b,b^2]}.
\begin{figure}[htbp]
\begin{center}
\includegraphics[width=2 in]{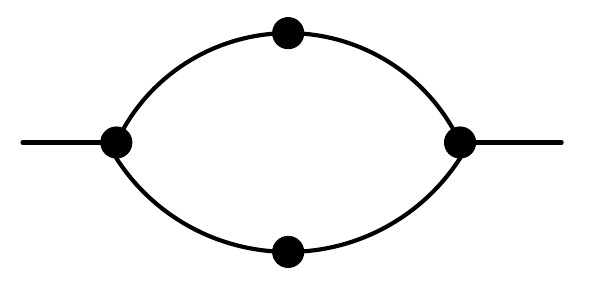} \caption{The
function on degree 2 nodes is $[1,a,-1]$, and the function on degree
3 nodes is $[1,0,0,b]$. } \label{Figure:[1,a^2b,b^2]}
\end{center}
\end{figure}
Its signature is $[1,a^2b, b^2]$, and by Lemma
\ref{lemma:interpolation}, we are done. This completes the proof.
\end{proof}

By the above lemmas, the only case left we have to handle is that ${\cal
F} \subseteq {\cal G}_1 \cup  {\cal F}_3$. This is done by the
following lemma, which completes the proof of Theorem
\ref{thm:holant-c}.

\begin{lemma}\label{lemma:G1F3}
If ${\cal F}  \subseteq {\cal G}_1 \cup  {\cal F}_3$, then the conclusions
of
Theorem
\ref{thm:holant-c} hold.
\end{lemma}

\begin{proof}
If ${\cal F} \subseteq {\cal F}_1 \cup {\cal F}_3$,
then by  Theorem~\ref{thm-holnat-c} part (2), ${\rm
Holant}^c({\cal F})$ is computable in polynomial time.
Similarly, if  ${\cal F} \subseteq
 {\cal U} \cup {\cal F}_3 \cup \{[1,0,1]\}$, then
 by  Theorem~\ref{thm-holnat-c} part (1), and then by
Theorem~\ref{thm-holant-star} part (3),
${\rm Holant}^c({\cal F})$ is computable in polynomial time.
Hence in these two cases, Theorem~\ref{thm:holant-c} holds.
Now suppose ${\cal F}  \not \subseteq {\cal F}_1 \cup {\cal F}_3$
and ${\cal F} \not \subseteq
{\cal U} \cup {\cal F}_3 \cup \{[1,0,1]\}$.

There exists $f \in {\cal F} - {\cal F}_1 \cup {\cal F}_3$.
Since ${\cal F}  \subseteq {\cal G}_1 \cup  {\cal F}_3$,
such an $f \in {\cal G}_1$.

Now there are two cases. The first case is that
we have such an $f \not \in {\cal U}$,
and so, $f \in {\cal F} \cap {\cal G}_1 -  ({\cal F}_1 \cup {\cal F}_3
\cup {\cal U})$.
The arity of $f$ is greater than $1$.
By connecting its dangling edges together except two or three depends
on the parity of the arity of $f$,
we can assume $f$ has form
$[1,0,a]$ or $[1,0,0,a]$, where $a \not \in \{-1,0,1\}$.

The second case is  every $f \in
{\cal F} \cap {\cal G}_1 -  ({\cal F}_1 \cup {\cal F}_3)$ is also in ${\cal U}$.
By $ {\cal F} \not \subseteq {\cal U} \cup {\cal F}_3 \cup \{[1,0,1]\}$,
there exists $f_1 \in {\cal F} - ({\cal U} \cup {\cal F}_3 \cup \{[1,0,1]\})$.
Since ${\cal F} \subseteq {\cal G}_1 \cup  {\cal F}_3$,
and $f_1 \not \in {\cal F}_3$, we get $f_1 \in {\cal G}_1$.
If $f_1  \not \in {\cal F}_1$, we could use this $f_1$ as the $f$ above,
namely $f_1 \in {\cal F} \cap {\cal G}_1 -  ({\cal F}_1 \cup {\cal F}_3
\cup {\cal U})$.
A contradiction.
Thus $f_1  \in {\cal F}_1$.
Also we have some $f_2 \in {\cal F} - ({\cal F}_1 \cup {\cal F}_3)$.
So $f_2 \in {\cal G}_1$, since ${\cal F} \subseteq {\cal G}_1 \cup
{\cal F}_3$.
Also since we are in this second case, certainly $f_2 \in {\cal U}$.


So we have
$f_1, f_2\in {\cal F} \cap {\cal G}_1$ such that $f_1 \in {\cal F}_1$ but
$f_1 \not \in{\cal U} \cup {\cal F}_3 \cup \{[1,0,1]\}$, and $f_2
\in {\cal U}$ but $f_2 \not \in {\cal F}_1$.
The arity of $f_1$ is at least 2. We claim it is greater than 2.
 Otherwise,  $f_1$ being from ${\cal F}_1$ and not $[1,0,1]$,
it would be $f_1 = [1,0,-1] \in {\cal F}_3$, a contradiction.
So $f_1$ has form $[1,0,0,\ldots,\pm 1]$ of arity at least 3.  $f_2$
is of form $[1,a']$, where $a' \not \in \{-1,0,1\}$;
this follows from $f_2 \in {\cal U} \cap {\cal G}_1 - {\cal F}_1$.
By connecting
all the dangling edges of $f_1$ except two with $f_2$, we can
construct an \FF ~with signature of form $[1,0,a]$, where $a \not \in
\{-1,0,1\}$. This is one of the above two forms after the first case.
To sum up, in both cases, we have some
$f$ of the form
$[1,0,a]$ or $[1,0,0,a]$, where $a \not \in \{-1,0,1\}$.

If ${\cal F}  \subseteq {\cal G}_1 \cup  \{[0,1,0]\} \cup {\cal U}$,
then by Theorem~\ref{thm-holnat-c} part (1),
and then by Theorem~\ref{thm-holant-star} part (2) (with $a=0$ and $b=1$),
 ${\rm Holant}^c
({\cal F})$ is computable in polynomial time and Theorem
\ref{thm:holant-c} holds. Otherwise, there exists $g\in {\cal F}
\subseteq {\cal G}_1 \cup  {\cal F}_3 $, and $g \not \in {\cal G}_1
\cup  \{[0,1,0]\} \cup {\cal U}$. Then $g$ must be in ${\cal F}_3$, and
 have one of the following
sub-signatures: $[1,1,-1],[1,-1,-1],[1,0,-1,0],[0,1,0,-1]$;
this follows from a careful examination of the forms of ${\cal F}_3$. By
symmetry (taking the reversal of both $f$ and $g$),
we only need to consider the cases $f = [1,0,a]$ or $[1,0,0,a]$, where $a
\not \in \{-1, 0, 1\}$,
and $g = [1,1,-1]$ or
$[1,0,-1,0]$.

According to $f$ and $g$, we have four cases. If $f=[1,0,a]$ and
$g=[1,1,-1]$, then connecting them together into a chain $fgf$, we
can realize $[1,a,-a^2]$. By Lemma \ref{lemma:interpolation}, we are
done. If $f=[1,0,a]$ and $g=[1,0,-1,0]$, for each dangling edge
of $g$, we extend it by one copy of $f$. Then we can realize
$[1,0,-a^2,0]$. So by Lemma \ref{lemma:010x}, we are done. If $f=[1,0,0, a]$ and $g=[1,1,-1]$, we can connect a unary
signature $[1,1]$ (sub-signature of $g$) to one dangling edge of
$f$, and realize a binary signature $f=[1,0,a]$. This reduces it to the
first case, which has been proved. If $f=[1,0,0,a]$ and
$g=[1,0,-1,0]$, we can realize a unary signature $[1,a]$ from $f$
by connecting two of its dangling edges together,
and then connect this unary signature to one dangling edge of $g$ to
realize $[1,-a,-1]$. Note that $[1,-a,-1]\not \in {\cal G}_1 \cup
{\cal F}_3$, by Lemma \ref{lemma:inG1F3}, we are done.
\end{proof}

%% file: csp.tex
\section{Dichotomy for Planar Weighted \#CSP}\label{section:csp}
In this section, we prove a dichotomy for planar real weighted \#CSP.
Compared to the dichotomy for general real weighted \#CSP, the new tractable
cases for planar structures
 are precisely
 those which can be computed by holographic algorithms with
matchgates. Since all the equality functions are assumed to be available, the only possible
basis used in holographic algorithms is  $\begin{bmatrix}1 & 1 \\ 1 & -1 \end{bmatrix}$ (this can be computed by
the characterization in \cite{STOC07}).
Now we present the dichotomy theorem for   planar weighted \#CSP.

\begin{theorem}\label{thm:csp}
Let $\mathcal{F}$ be a set of real symmetric functions.
{\rm Pl}-\#{\rm CSP}$(\mathcal{F})$ is \#P-hard unless $\mathcal{F}$
satisfies one of the following conditions, in which case it is
tractable:
\begin{enumerate}
    \item ${\rm \#CSP}(\mathcal{F})$ is tractable (for which
we have an effective dichotomy~\cite{STOC09}); or
    \item Every function in $\mathcal{F}$ is realizable by some matchgate under basis $\begin{bmatrix}1 & 1 \\ 1 & -1 \end{bmatrix}$
(for which we have a complete characterization \cite{MGI}).
\end{enumerate}
\end{theorem}

The main proof idea is to reduce  Pl-Holant$^c$ problems to Pl-\#CSP
problems.
Pl-\#CSP$(\mathcal{F})$ is  exactly the same  as
planar Holant  with all the {\sc Equality} functions, i.e.,
Pl-Holant$(\mathcal{F}\cup \{[1,1], [1,0,1], [1,0,0,1],[1,0,0,0,1], \ldots\})$.
We can use a holographic reduction under the basis  $H=\begin{bmatrix}1 & 1 \\ 1 & -1 \end{bmatrix}$.
Under this transformation,
the problem is transformed to, and hence has
 the same complexity as Pl-Holant$(H \mathcal{F}\cup \{[1,0], [1,0,1], [1,0,1,0],[1,0,1,0,1], \ldots\})$.
Since this holographic reduction gives us $[1,0]$ (from $[1,1]$),
if we can further realize (or interpolate) $[0,1]$,
we can view the problem as a Pl-Holant$^c$ problem and
apply Theorem~\ref{thm:holant-c} to $H \mathcal{F}\cup \{[1,0,1], [1,0,1,0],
[1,0,1,0,1], \ldots\}$
to get  a  proof of Theorem~\ref{thm:csp}.
 In the following, we show how to  realize (or interpolate) $[0,1]$.
Once we have  $[0,1]$,
the translation of the criterion of  Theorem~\ref{thm:holant-c}
to  Theorem~\ref{thm:csp} is straightforward.

It turns out that to  realize (or interpolate) $[0,1]$
in some cases is difficult.
The following lemma says that it is also sufficient if we
can realize (or interpolate) $[0,0,1]$.
$[0,0,1]$ can be viewed as two copies of $[0,1]$,
as $[0,0,1]=[0,1] \otimes [0,1]$.
Intuitively, we will use one copy of  $[0,0,1]$
to replace two occurrences of $[0,1]$.
 However, there are two technical
difficulties. One is that there may be an odd number of occurrences
 of $[0,1]$ used in the input instance; the second
difficulty, which is more subtle, is that we have
to pair up two copies of $[0,1]$ while maintaining
planarity of the instance.

\begin{lemma}
Pl-Holant$(\mathcal{F}\cup \{[1,0],[0,0,1], [1,0,1], [1,0,1,0],[1,0,1,0,1],
 \ldots\})$ is \#P-hard (or in P) if and only if
Pl-Holant$^c(\mathcal{F}\cup \{[0,0,1], [1,0,1], [1,0,1,0],[1,0,1,0,1],
 \ldots\})$  is \#P-hard (or in P).
\end{lemma}

\begin{proof}
There is one more function $[0,1]$ in the second signature set
 than the first, so obviously the first one can be reduced to the second one.
Hence if the second problem is in P, so is the first.
We have already  proved a dichotomy theorem for Pl-Holant$^c$ problems.
So now we may assume the second problem is \#P-hard,
and show that the first problem is also \#P-hard.

We observe that all the proofs in this paper and \cite{STOC09}, when the second problem
for any signature set  is proved to be \#P-hard, one of the following three problems:
 (a) Pl-Holant$([1,0,0,1]|[1,1,0])$, (b) Pl-Holant$([1,1,0,0])$, or (c) Holant$[0,1,0,0]$
(respectively counting {\sc Vertex Cover},  {\sc Matching}
for planar 3-regular  graphs,
or {\sc Perfect Matching} for general 3-regular graphs)
 is reduced to it by a chain of reductions.
There are only three reduction methods in this reduction chain, direct gadget construction, polynomial interpolation, and holographic reduction.

Given an instance $G$ of Pl-Holant$([1,0,0,1]|[1,1,0])$,
 Pl-Holant$([1,1,0,0])$, or Holant$[0,1,0,0]$, we consider the
graph $G\cup G$, which denotes the disjoint union of two
copies of  $G$.

Notice that
 the value of Pl-Holant$([1,0,0,1]|[1,1,0])$,
 Pl-Holant$([1,1,0,0])$, or Holant$[0,1,0,0]$ on the instance $G$ is  a non-negative integer, and the
value on  $G\cup G$ is its square.
 So we can compute the value on $G$
 uniquely  from its square.
Suppose the reduction chain on the instance $G$ produced
instances $G_1,G_2,\ldots, G_m$ of the second problem.
The same reduction applied to   $G \cup G$ produces
 instances of the form $G_1 \cup G_1, G_2 \cup G_2, \ldots, G_{m'}
\cup G_{m'}$. (We note that the reduction on  $G \cup G$
may  produce polynomially more instances than on  $G$
because of polynomial interpolation.)


Now we only need to show how to transform
  instances $G_1 \cup G_1,G_2 \cup G_2,\ldots, G_{m'} \cup G_{m'}$
in the second problem,
to instances of the first problem with the same values
(replacing all occurrences of the signature  $[0,1]$
 by some  $[0,0,1]$).
 $G_i \cup G_i$ is a planar graph with zero or more vertices of degree one
attached with the function $[0,1]$. We want to use one copy of $[0,0,1]$
to replace
one  pair of  $[0,1]$, while maintaining  planarity.

Take a spanning tree of the dual graph of $G_i$. Let the outer face be
 the root. Choose an arbitrary leaf of this tree, which corresponds to a face $C$ of $G_i$. Suppose $C'$ is the face corresponding to the parent
 of $C$ in the tree. If there are an even number of vertices of
degree one attached with  $[0,1]$ in face $C$,
we can perfectly match them and realize them using $[0,0,1]$ while
maintaining planarity in this face.
This can be done by matching these dangling vertices of degree one
in a clockwise fashion on this face $C$.
 If there are an odd number of  $[0,1]$ in face $C$,
we choose one edge  $e$ between $C$ and $C'$, and
add a new vertex $v_e$ on $e$, and connect two new vertices
of degree one to  $v_e$.
The two new vertices are attached $[0,1]$, and
 $v_e$ has degree 4 and is attached $[1,0,1,0,1]$.
The effect of  $[1,0,1,0,1]$ connected by two $[0,1]$ is the same
as the function $[1,0,1]$, which is exactly the same as
the  edge  $e$ itself.
 We put one new vertex with $[0,1]$ in face $C$, and the other
 one in face $C'$. Now, there are an even number of $[0,1]$  in
face $C$, and we can replace them by  $[0,0,1]$ in $C$,
as before.
We may repeat this process, until we reach the root in
the dual graph of $G_i$. If we do the same
for the two $G_i$ in $G_i \cup G_i$, we will
have an even   number of $[0,1]$
in the common outer face and can at last perfectly match the $[0,1]$ vertices
and realize them by $[0,0,1]$. In the end
we get an instance of the first problem, which has the same value.
\end{proof}

To sum up the above discussion, and apply Theorem~\ref{thm:holant-c},
 we have the following lemma, which is the starting point of our proof
of Theorem~\ref{thm:csp}.

\begin{lemma} \label{lemma001}
If we can realize (or interpolate) $[0,1]$ or $[0,0,1]$ from  $H \mathcal{F}\cup \{[1,0], [1,0,1], [1,0,1,0],[1,0,1,0,1], \ldots\}$,
then the conclusion of Theorem \ref{thm:csp} holds.
\end{lemma}

Next we give two lemmas which give a general condition to
 realize or interpolate  $[0,1]$ or $[0,0,1]$.

\begin{lemma}\label{lem:get [0,1] interpol}
Let $a \in \mathbb{R}$.
If $a \not \in \{0,1,-1\}$, then
we can interpolate $[0,1]$
from $(\{[1,a], [1,0], [1,0,1], [1,0,1,0],
\ldots\})$.
\end{lemma}

\begin{proof}
For every $j \ge 1$, we can take a function
$F_{j+1}=[1,0,1,0,1,\ldots]$  of arity $j+1$,
 and connect $j$ functions $[1,a]$ to
it.
The row vector form
  of the function (i.e., a listing of its values)
of arity $j$ composed of  $j$
copies of  $[1,a]$ is $(1,a)^{\otimes j}$.

The column vector form of $F_{j+1}$ is $1/2 \left( \begin{array}{c}
1 \\ 1 \end{array} \right)^{\otimes (j+1)}+1/2 \left( \begin{array}{c}
1 \\ -1 \end{array} \right)^{\otimes (j+1)}$. The $2^j \times 2$
matrix form of $F_{j+1}$ is $1/2 \left( \begin{array}{c} 1 \\ 1
\end{array}
\right)^{\otimes j} \otimes (1, 1)+1/2 \left( \begin{array}{c} 1 \\
-1
\end{array} \right)^{\otimes j} \otimes (1, -1)$.

Our gadget realizes
\[ (1,a)^{\otimes j}
\left[
1/2 \left( \begin{array}{c} 1 \\ 1 \end{array}
\right)^{\otimes j} \otimes (1, 1)+1/2 \left( \begin{array}{c} 1 \\
-1
\end{array} \right)^{\otimes j} \otimes (1, -1)
\right]
=\frac{(1+a)^{j}}{2} (1,  1)  + \frac{(1-a)^j}{2}  (1, -1) \textrm{. } \]

Because $a \in  \mathbb{R}$ and
$a  \not \in \{0,1,-1\}$, $(1+a)/(1-a)$ is well defined
and is neither zero nor  a root of unity.
We can interpolate any unary function $x(1,  1)  + y (1, -1)$,
in particular $[0,1]$.
\end{proof}

\begin{lemma}\label{lem:get [0,0,1] interpol}
Let $a \in \mathbb{R}$.
If $a \not \in \{0,1,-1\}$, then
we can interpolate $[0,0,1]$ from
$[1,0,a]$.
\end{lemma}

\begin{proof}
The function of a chain of length $j$ composed of $[1,0,a]$ is
$[1,0,a^j]$. Since the real number $a \not \in \{0,1,-1\}$, we can
interpolate all $[x,0,y]$, and in particular $[0,0,1]$,
by polynomial interpolation.
\end{proof}


\noindent {\bf Proof of Theorem \ref{thm:csp}.}
In this proof, we augment the class
$\mathcal{F}_1 \cup \mathcal{F}_2 \cup \mathcal{F}_3$
to include  those degenerate signatures which can be
obtained from tensor products from unary signatures
in $\mathcal{F}_1 \cup \mathcal{F}_2 \cup \mathcal{F}_3$.

If $H\mathcal{F} \subseteq \mathcal{F}_1 \cup \mathcal{F}_2 \cup \mathcal{F}_3 $, then
the problem is tractable even for general graphs
 and the conclusion of the theorem holds.
Now we assume that there exists an $f \in H\mathcal{F} - (\mathcal{F}_1 \cup \mathcal{F}_2 \cup \mathcal{F}_3)$.
In the following, we will prove
 that we can realize (or interpolate) $[0,1]$ or $[0,0,1]$ from
$f$ and $\{[1,0], [1,0,1], [1,0,1,0],[1,0,1,0,1],\ldots\}$.
%

The general thrust of the proof is to squeeze all
possible $f$ into several standardized forms,
and either prove \#P-hardness or reach a contradiction.
 We assume for a
contradiction that we cannot realize
(or interpolate) $[0,1]$ or $[0,0,1]$.
Suppose
$f=[f_0,f_1,\ldots,f_n]$. Since we have $[1,0]$, we can always  take
an initial subsequence of an $f$ we already have
 as the signature of
a realizable function. Given a
function $g$ with arity $r>1$, we often use the gadget composed of
two  copies of $g$ such that $r-1$ inputs of them are connected to each other.
We call this the double gadget from $g$.
We separate two cases according to whether $f_0=0$, or $f_0 \not =0$
which we normalize to $f_0=1$.

\begin{enumerate}
\item{$f_0=0$.}

As the constant 0 function is in $\mathcal{F}_1 \cup \mathcal{F}_2
\cup \mathcal{F}_3$,  $f$ is not identically 0, and thus for some
$i\ge 1$, $f_i  \not = 0$. If $f_0=0$ and $f_1 \neq 0$, then we can
connect  $n-1$ functions $[1,0]$ to $f$ to get $[0,f_1]$, which is
$[0,1]$ up to a factor.

So we have  $f_0=f_1 = 0$, then  $n \ge 2$.
If  $f_2 \neq 0$, then we can connect  $n-2$
functions $[1,0]$ to $f$ to get $[0,0,f_2]$, which is $[0,0,1]$ up
to a factor.

So we have  $f_0=f_1 =f_2 = 0$, then  $n \ge 3$.
Let $m \le n$ be the first nonzero,
$f_0=f_1=f_2=\cdots=f_{m-1}=0$, $f_m \neq 0$, then we can
connect a function $[1,0,1]$ to two dangling edges of
 $f$ to get a function whose first
nonzero entry is $f_m$ at index $m-2$.
We can repeat this process until exactly one or two zeros are
left at index 0 or  at index 0 and 1,  and
we reach one of the two scenarios above.

\item{$f_0=1$.}

By Lemma
\ref{lem:get [0,1] interpol}, we only need to consider $f_1  \in
\{0,1,-1\}$. Otherwise, we are done.

\begin{enumerate}
\item{$f_0=1$ and $f_1=\pm 1$.}


If $n=1$, then $f =
[f_0,f_1]=[1,\pm 1]\in \mathcal{F}_1 \cup \mathcal{F}_2 \cup \mathcal{F}_3$, a contradiction.
Therefore we have $n\geq 2$, we can take its initial part $[1,f_1,f_2]$.
Connecting one edge to  $[1,f_1]$, we get $[1+f_1^2,f_1+f_1f_2]
=[2, \pm (1+ f_2)]$. By Lemma
\ref{lem:get [0,1] interpol}, we only need to consider $f_2  \in
\{1,-1,-3\}$.

We can construct another gadget which connects two inputs of $[1,0,1,0]$ by
$[1,f_1,f_2]$. This produces a unary signature  $[1+f_2,2f_1]$.
It follows that $f_2 \neq -1$, since otherwise we have $[0, 1]$
after normalization.

Next we rule out $f_2=-3$.
The  double gadget of $[f_0, f_1,f_2]
= [1, \pm 1, f_2]$ has signature  $[2,-2,10]$ and
$[2,2,10]$. After normalizing,
this gives $[1, \pm 1, 5]$ and $5 \not \in \{1,-1,-3\}$.
 Hence, we may assume  $f_2=1$.

Our goal in this case 2.(a) of $f_0=1$ and $f_1=\pm 1$
is to extend this pattern $[1, \pm 1, 1, \ldots]$.
Assume we have proved that $f_j=1$ (respectively $f_j=(-1)^j$) for
$j=0,1,\ldots,m$ ($m \geq 2$). Since $f \not \in
 \mathcal{F}_1 \cup \mathcal{F}_2 \cup \mathcal{F}_3$,
the arity $n > m$. We can connect $[f_0,f_1]$ to
$[f_0,f_1,\ldots, f_{m+1}]$, to get a function
$[f_0^2 + f_1^2, f_0 f_1 + f_1 f_2, f_0 f_2 + f_1 f_3, \ldots]$
 of arity $m$, which is
$[2,2,\ldots,1+f_{m+1}]$
 (respectively $[2,-2,\ldots,f_{m-1}-f_{m}
,f_m-f_{m+1}]$).
By what has been proved inductively,
$f_{m+1}=1$ (respectively $f_{m+1}=f_{m-1}$).
So in this case
 we showed that either
 $f \in \mathcal{F}_1 \cup \mathcal{F}_2 \cup \mathcal{F}_3$,
which is a contradiction, or \#P-hardness.

\item{$f_0=1$ and $f_1=0$.}

Since $[1,0] \in \mathcal{F}_2$, and $f \not \in \mathcal{F}_2$,
 we have $n > 1$.
If $f$ were degenerate it would be $[1,0]^{\otimes n} = [1, 0,
\ldots, 0]$, which would  belong to the augmented class of
$\mathcal{F}_1 \cup \mathcal{F}_2 \cup \mathcal{F}_2$.  But $f$ does
not. So $f$ is non-degenerate, in particular there is some nonzero
entry other than $f_0$. Suppose $f_m$ is the first nonzero $f_i$,
for $i > 1$.  Because we can connect some $[1,0,1]$ to
$[f_0,0,\ldots,0, f_m]$ to get $[1,f_m]$ or $[1,0,f_m]$, we have
$f_m=\pm 1$ by Lemma~\ref{lem:get [0,1] interpol} and \ref{lem:get
[0,0,1] interpol}. Since $f \not \in \mathcal{F}_1$ we have $n > m$.

We prove $m$ is an even number. Otherwise, we can get
$[1,f_m,f_{m+1}]$ by connecting some $[1,0,1]$. By the proof for
case 2.(a),
we get  $f_{m+1}=1$. We can also get $[1,0,0,f_m,f_{m+1}]$, since $m
> 1$, whose double gadget has the signature $[1+f_m^2,f_m
f_{m+1},3f_m^2+f_{m+1}^2]=[2,\pm 1,4]$. This gives \#P-hardness, as
we can get $[2, \pm 1]$.

Now we know $m$ must be even. Next we show that in fact $m=2$.
Otherwise, $m \ge 4$ and we can
get $[1,0,0,0,f_m,f_{m+1}]$, whose double gadget has the signature
$[1+f_m^2,f_m f_{m+1} ,4f_m^2+f_{m+1}^2]=[2,\pm f_{m+1},4+f_{m+1}^2]$.
By what has been proved so far
this also leads to \#P-hardness.

Now we have  $m=2$ and  reached $[1, 0, \pm 1, f_3]$, whose double
gadget has the signature $[2, \pm f_3, 2+f_3^2]$, so $f_3=0$.

Again since $f \not \in \mathcal{F}_1  \cup \mathcal{F}_2 \cup
\mathcal{F}_3$ we have $n > 3$. Hence we have $[1,0,\pm 1, 0, f_4]$.
For $[1,0,-1, 0, f_4]$, by connecting two edges with $[1,0,1]$, we
get $[0,0,-1+f_4]$, and we must have $f_4=1$, or else we have the
signature $[0,0,1]$,  after normalization. For $[1,0,1, 0, f_4]$, by
connecting two edges with $[1,0,1]$, we get $[2,0, 1+f_4]$, and it
follows from Lemma~\ref{lem:get [0,0,1] interpol}
 that  $f_4 \in \{1,-1,-3\}$.
Connecting three edges of $[1,0,1, 0, f_4]$ to three edges of $[1,0,1, 0, 1]$,
we get $[4,0,3+f_4]$, which rules out $f_4=-1$,
by Lemma~\ref{lem:get [0,0,1] interpol} again.
 The double gadget of $[1,0,1, 0, f_4]$
gives $[4,0,3+f^2_4]$, which rules out $f_4=-3$. To sum up, we get $f_4=1$.

We have  reached $[1,0,\pm 1, 0, 1, \ldots]$.
The rest of the proof is
similar to the induction proof for
the case 2.(a)
but by skipping all
entries with an odd index.
Assume we have proved that $f_j$ are of the proper
form, for $j=0,1,\ldots,m$. More precisely,
$f_j = 0$ for all odd $j \le m$, and,
either $f_{2j} = 1$ for all $j \le m/2$, or $f_{2j} = (-1)^j$
for all $j \le m/2$.  Since  $f \not \in \mathcal{F}_1  \cup \mathcal{F}_2 \cup \mathcal{F}_3$,
the arity $n > m$.  We
connect $[f_0,f_1,f_2]$ to $[f_0,f_1,\ldots, f_{m+1}]$, to get a
function $g$ of arity $m-1$. If $m$ is even,
$f_{m-1}=0$, and  $f_{m+1}$ is added to or subtracted from
$f_{m-1}$, namely $f_{m-1} \pm f_{m+1}$,
to form the  last entry in $g$ at index $m-1$.
This entry
should be zero by induction, so $f_{m+1}=0$.
If $m$ is odd, we can repeat the  proof in the case 2.(a),
but   we ignore all zero entries at
odd indexed locations,
then the induction can be
completed as before.
This completes the proof.
\end{enumerate}
\end{enumerate}

\qed

%% file: 2-3_regular.tex
\section{Dichotomy for Planar 2-3 Regular Graphs}\label{section:2-3-regular}
In this section we prove a dichotomy for Holant on planar 2-3 regular graphs.
This setting is very interesting for at least two reasons. From dichotomy theorem point of
view, this is the simplest nontrivial setting and always serves
as the starting point of more
general dichotomy theorems as in \cite{STOC09,CHL09}.  This was also a focus of several previous work \cite{FOCS08,mike09,TAMC09,cai},
whose result is the starting point of this theorem.
From the holographic algorithms point of view, most of the known holographic algorithms \cite{HA_J,AA_FOCS}
are essentially for planar 2-3 regular graphs. The dichotomy theorem here explains the reason why they
are special and why many variations of them are \#P-hard.
In the previous two dichotomies for Pl-Holant$^c$ and Pl-\#CSP, the new tractable cases for planar are also done by
holographic algorithms with matchgates. However, only special basis transformations are used since we assume some
signatures are freely available. In this planar 2-3 regular graphs setting, no additional signatures are assumed to
be freely available. Therefore all possible bases can be used in tractable cases.

 \begin{theorem}\label{thm:2-3-regular}
Let $[y_0,y_1, y_2]$ and $[x_0,x_1,x_2,x_3]$ be two complex symmetric signatures with
arity 2 and 3 respectively. Then {\rm Pl-Holant}$([y_0,y_1, y_2]|[x_0,x_1,x_2,x_3])$
is \#P-hard unless  $[y_0,y_1, y_2]$ and $[x_0,x_1,x_2,x_3]$
satisfy one of the following conditions, in which case it is
tractable:
\begin{enumerate}
    \item {\rm Holant}$([y_0,y_1, y_2]|[x_0,x_1,x_2,x_3])$ is tractable (for which
we have an effective dichotomy~\cite{CHL09}); or
    \item There exists a basis $T$ such that both $[y_0,y_1, y_2] (T^{-1})^{\otimes 2}$ and $T^{\otimes 3}[x_0,x_1,x_2,x_3]$ are  realizable by some matchgates
(for which we have a complete characterization \cite{STOC07}).
\end{enumerate}
\end{theorem}

\begin{proof}
If $[x_0,x_1,x_2,x_3]$ or $[y_0,y_1, y_2]$ is degenerate, the problem is tractable, even for the non-planar case, and so
this falls in condition {\it 1}. Now we assume that
they are both non-degenerate.
As proved in \cite{STOC09}, we can choose an  invertible $T_1$ such that
$[x_0,x_1,x_2,x_3]$ (or its reversal, which is similar and we omit that case) can be written
as $T_1^{\otimes 3} [1,0,0,1]$ or $T_1^{\otimes 3} [1,1,0,0]$. Therefore by
a holographic reduction,
we can always reduce the problem equivalently to one of the following two problems:
(1) {\rm Pl-Holant}$([z_0,z_1, z_2]|[1,0,0,1])$ and (2) {\rm Pl-Holant}$([z_0,z_1, z_2]|[1,1,0,0])$.
So it is sufficient to prove the theorem for these two cases.

For {\rm Pl-Holant}$([z_0,z_1, z_2]|[1,0,0,1])$, by the theorem
\ref{lemma-cai} (\cite{cai}),
the only case which is hard for general graphs
 and tractable for planar graphs is
 $z_0^3=z_2^3$. This condition is exactly the same as the condition that there exists
 a basis $T$ such that both $[y_0,y_1, y_2] (T^{-1})^{\otimes 2}$ and $T^{\otimes 3}[1,0,0,1]$ are  realizable by some matchgates.
This proves Theorem~\ref{thm:2-3-regular} for case (1).

Now we consider {\rm Pl-Holant}$([z_0,z_1, z_2]|[1,1,0,0])$.
If $z_0=0$, the problem is trivially
tractable even for general graphs.
This can be seen by a simple counting argument:
in a bipartite graph the LHS vertices
all have the signature $[0, z_1, z_2]$ and thus at least
half the edges must be 1,
while the RHS vertices
all have the signature $[1,1,0,0]$ and thus less than
half  the edges are 1.
  This is also the only case where the problem is not \#P-hard
for general graphs
when the RHS has $[1,1,0,0]$ by~\cite{CHL09}.
Now we assume $z_0 \neq 0$. Then it is sufficient to prove that either the problem is \#P-hard or there
exists a basis transformation $T$ such that $[1,1,0,0]T^{\otimes 3}$ and  $ (T^{-1})^{\otimes 2} [z_0,z_1, z_2]$  are  realizable by some matchgates.
Let $T=\begin{bmatrix}\sqrt{z_0} & 0 \\
z_1/{\sqrt{z_0}}  &  \sqrt{(z_0 z_2 - (z_1)^2)/z_0}
\end{bmatrix}$.
Note that $T$ is well defined
and invertible since $z_0\neq 0$ and $[z_0, z_1, z_2]$ is non-degenerate (
i.e., $z_0 z_2 - (z_1)^2 \neq 0$).
Then we can verify that
$$[1,1,0,0]T^{\otimes 3}= [\sqrt{z_0}(z_0+3 z_1), \sqrt{z_0 (z_0 z_2 - (z_1)^2)}, 0,0] \mbox{ ~~~ and ~~~ } (T^{-1})^{\otimes 2} [z_0,z_1, z_2]=[1,0,1]. $$
We note that $\sqrt{z_0 (z_0 z_2 - (z_1)^2)}\neq 0$.
If $\sqrt{z_0}(z_0+3 z_1)=0$, then both $[\sqrt{z_0}(z_0+3 z_1), \sqrt{z_0 (z_0 z_2 - (z_1)^2)}, 0,0] $ and $[1,0,1]$ can be realized by
matchgates and the problem for planar graphs
 is tractable. We denote $v=\frac{\sqrt{z_0}(z_0+3 z_1)}{\sqrt{z_0 (z_0 z_2 - (z_1)^2)}}\neq 0$.
Then the problem is equivalent to (non-bipartite) Pl-Holant($[v,1,0,0]$).
Now it is sufficient to prove the following claim:

{\bf Claim:} Let $v\neq 0$ be a complex number. Then Pl-Holant($[v,1,0,0]$) is \#P-hard.

We can realize $[v^3+3v,v^2+1,v,1]$ by connecting
3 copies of
$[v,1,0,0]$'s as illustrated in Figure \ref{fig-g0}.
\begin{figure}[httb]
    \centering
    \includegraphics[height=2.6 cm]{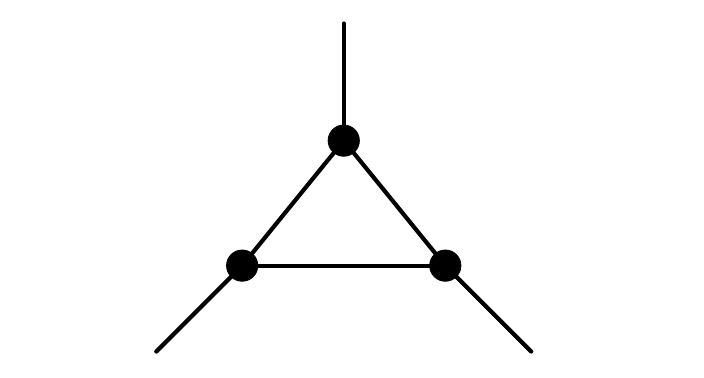}
    \caption{All vertex signatures are $[v,1,0,0]$.}
    \label{fig-g0}
\end{figure}
If we can prove that  Pl-Holant($[v^3+3v,v^2+1,v,1]$) is \#P-hard, then we are done.
 In tensor product notation this signature is
\begin{displaymath}
[v^3+3v,v^2+1,v,1]^{\tt T} = \frac{1}{2}\left(
    \left[
        \begin{array}{c}
            v+1 \\
            1
        \end{array}
    \right]^{\otimes 3}
    +\left[
        \begin{array}{c}
            v-1 \\
            1
        \end{array}
    \right]^{\otimes 3}
\right).
\end{displaymath}
Then the following reduction chain holds:
\begin{eqnarray*}
  {\mbox {\rm Pl-Holant}}([v^3+3v,v^2+1,v,1])
& \equiv_{\tt T} & {\mbox {\rm Pl-Holant}}([1,0,1]|[v^3+3v,v^2+1,v,1]) \\
& \equiv_{\tt T} & {\mbox {\rm Pl-Holant}}([v^2+2v+2,v^2,v^2-2v+2]|[1,0,0,1]),
\end{eqnarray*}
where the second step is a holographic reduction using
$\left[
        \begin{array}{c c}
                v+1 & v-1 \\
                1 & 1
        \end{array}
\right]$.
This transforms the problem to our first case
where the RHS all have $[1,0,0,1]$.
The only possible exceptional case happens when  $(v^2+2v+2)^3=(v^2-2v+2)^3$.
Since $ (v^2+2v+2)^3-(v^2-2v+2)^3 =4 v (3 v^4+ 16 v^2+ 12)$ and $v\neq 0$,
we will have proved the claim as long as $3 v^4+ 16 v^2+ 12 \neq 0$.
There are four roots for the equation $3 v^4+ 16 v^2+ 12 = 0$, and for these
four exceptional values of $v$, we prove it separately as follows.

In addition to the gadget in Figure~\ref{fig-g0},
we can construct a gadget in Figure~\ref{fig-g1} with a
 binary signature $[v^2+2, v, 1]$.
\begin{figure}[httb]
    \centering
    \includegraphics[height=2.6 cm]{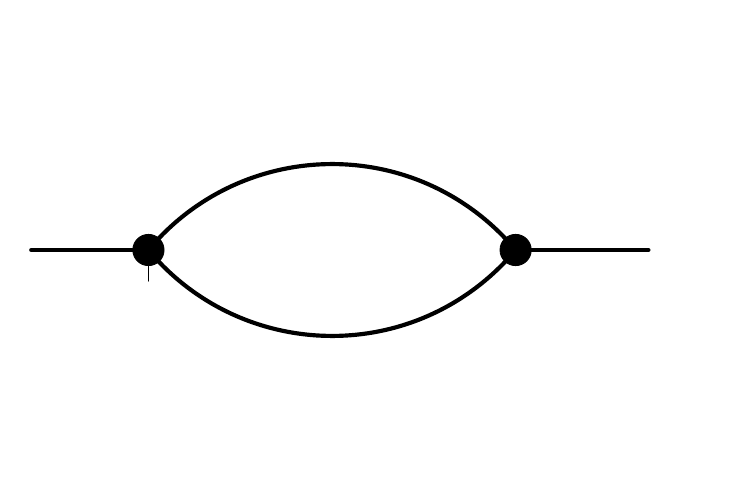}
    \caption{All vertex signatures are $[v,1,0,0]$.}
    \label{fig-g1}
\end{figure}
Now it is  enough to prove that Pl-Holant($[v^2+2,v,1]|[v^3+3v,v^2+1,v,1]$) is \#P-hard.
Under the same basis $\left[
        \begin{array}{c c}
                v+1 & v-1 \\
                1 & 1
        \end{array}
\right]$, we will get an equivalent problem
{\rm Pl-Holant}($[X, Y, Z] \mid [1,0,0,1]$),
where
$X = (v^2+2) (v^2+2v+1)+2 v(v+1)+ 1$,
$Y= (v^2+2)(v^2-1)+2 v^2 +1$,
and
$Z = (v^2+2) (v^2-2v+1)+2 v(v-1)+ 1$.
 Again this  transforms the problem to our first case,
and,  it is easy to verify that any
 root of  $3 v^4+ 16 v^2+ 12 = 0$ is not
a tractable case here.
 This completes the proof of the claim and also the proof of
 the theorem.
\end{proof}

%% file: physics.tex
\section*{Appendix: Some Connections to Statistical Physics}

In this section we describe some background and connections from
Statistical Physics. Our discussion is necessarily a superficial one,
both due to our limited knowledge and limitation on space.
The purpose is to illustrate that, even at such a superficial level,
 a strong
connection exists, and that our complexity results may shed some
light on the venerable question from physics: Exactly  what
``systems" can be solved ``exactly" and what ``systems" are
``difficult".

The Ising model was named after Ernst Ising~\cite{ising1925beitrag}.
Wilhelm Lenz invented this model and gave it to his student Ising
to work on it.  The model consists of a discrete set of variables,
called spins, that can be assigned one of two values (states).
These spins are usually placed on a lattice structure or a graph,
and each spin interacts with its nearest neighbors.

Denoting the values each spin $i$ can take as $\sigma_i = +1$
 and $-1$,
the energy (the Hamiltonian) of the Ising model is
$H(\sigma) = - \sum_{\mbox{edge} \{i,j\}} J_{i,j} \sigma_i \sigma_j$.
The interaction between spins $i$ and $j$
 is called ferromagnetic if $J_{i,j} >0$,
antiferromagnetic if $J_{i,j} < 0$,
and noninteracting  if $J_{i,j} =0$.
E.g.,
if all the spins are placed on a one-dimensional lattice,
then the antiferromagnetic one-dimensional Ising model (with
the same value $J_{i,j} = J <0$) has the energy
function $H = \sum_{i} \sigma_i \sigma_{i+1}$, after normalization.
The ferromagnetic two-dimensional Ising model
on a square lattice (with the same value $J_{i,j} = J>0$) has energy
$H = -\sum_{i,j}
(\sigma_{i,j} \sigma_{i, j+1} + \sigma_{i,j} \sigma_{i+1, j})$.
The Ising model may be modified by magnetic fields which amounts
to a unary function at each spin
$H = - \sum_{\mbox{edge} \{i,j\}} J_{i,j} \sigma_i \sigma_j
 - \sum_{i} h_i \sigma_i$.

The model is a statistical model. The  central premise of statistical physics
is that the probability of each
configuration $\sigma$ is given
 by the Boltzmann distribution,
$e^{- H(\sigma)/k T}/\sum_{\sigma} e^{- H(\sigma)/k T}$,
where $k$ is Boltzmann constant and $T$ is the (absolute) temperature.
This focuses attention on the partition function
\[Z = \sum_{\sigma} e^{- H(\sigma)/k T}.\]
Note that the exponential $e^{- H(\sigma)/k T}$ turns
this into a sum-of-product functions exactly as we discussed in \#CSP.

In 1925, Ising solved the one-dimensional Ising model~\cite{ising1925beitrag}.
The 2-dimensional square lattice Ising model with zero magnetic field
was solved by Onsager  in 1944~\cite{onsager1944crystal}.
Onsager announced the formula for the spontaneous magnetization
for the two-dimensional model in 1949 but did not give a derivation~\cite{onsager1949}.
C.N.Yang (1952) gave the first published proof of this formula~\cite{yang1952spontaneous},
using a limit formula for Fredholm determinants,
proved in 1951 by Szeg\"{o} in direct response to Onsager's work.
There are many extensions to the basic Ising model~\cite{mccoy1973two,baxter1982exactly}.

Another landmark achievement is the exact computation of
the number of perfect matchings (dimer problem) on any
planar graph using Pfaffians.
This was independently discovered by Kasteleyn and
by Fisher and Temperley~\cite{TF1961,Kasteleyn1961,Kasteleyn1967}.
This problem can also be nicely expressed by a partition function
in our Holant framework; where this time the Boolean variables
 are the edges (to include an edge or not), and the local constraint
function at each vertex is the {\sc Exact-One} function.
Freedman, Lov\'{a}sz and Schrijver~\cite{freedman-l-s}
recently proved that this partition function {\it cannot} be expressed
as a graph homomorphism function, where the vertices are
variables as in the Ising model.
However in the framework of Holant problems we can find a unity
for all these problems.

We note the following.  In the paper \cite{STOC07} we gave a
complete characterization of matchgate realizable symmetric
signatures. The following lemma is proved~\cite{STOC07}:

\begin{lemma} \label{basis:rec3}
The set of bases under which the signature $[x_0,x_1,x_2]$ is
realizable as a signature by some matchgate is
\[\left\{\left[
\begin{pmatrix} n_0 \\ n_1\end{pmatrix},
          \begin{pmatrix} p_0 \\ p_1 \end{pmatrix}  \right]
\in {\bf GL}_2({\mathbb{C}}) \left|~ \begin{matrix}
          x_0 p_1^2-2x_1 p_1 n_1+ x_2 n_1^2=0,x_0 p_0^2-2x_1 p_0 n_0+ x_2 n_0^2=0
          \\ \mbox{{\rm or}~~~}
x_0 p_0 p_1-x_1 (n_0 p_1+ n_1 p_0)+x_2 n_0 n_1=0 \end{matrix}  \right. \right\}  .\]
\end{lemma}

This has the consequence that under the basis
$\left[
\begin{pmatrix} n_0 \\ n_1\end{pmatrix},
          \begin{pmatrix} p_0 \\ p_1 \end{pmatrix}  \right]
= \left[
\begin{pmatrix} 1 \\ 1\end{pmatrix},
          \begin{pmatrix} 1 \\ -1 \end{pmatrix}  \right]$,
the signature $[x, y, x]$ is realizable by a matchgate, for all
values $x$ and $y$.  In terms of the Ising model,
when two interacting  spins $i$ and $j$ take the same
assignment value $\sigma_i = \sigma_j = \pm 1$,
the contribution to the Hamiltonian is $- J_{i,j}$,
and when they take the  opposite assignment $\sigma_i = - \sigma_j = \pm 1$,
the contribution is $J_{i,j}$. Translating this to
the contributions to the partition function we get
exactly the local constraint evaluation $x = e^{J_{i,j}/kT}$ when
inputs are $00$ or $11$,  and
$y = e^{-J_{i,j}/kT}$ when inputs are $01$ and $10$.

Then,
the theory of Holographic Algorithms tells us that
for planar graphs, this Ising model is exactly solvable
by a holographic reduction to the FKT algorithm.

The present paper, especially Theorem~\ref{thm:csp},
tells us why this is exactly where physicists stopped,
and attempts to generalize this to non-planar systems
have not been successful in the past 85 years.

Sorin Istrail~\cite{Istrail00} showed that computing the free energy
of an arbitrary subgraph of an Ising model on a lattice of dimension
three or more is NP-hard; see a nice article by Barry Cipra in the
SIAM News~\cite{Cipra2000}. A very partial list of a
 great deal of research on this and related models, from
a computational complexity perspective,  can be found in
~\cite{Bulatov08,BulatovD03,BulatovG04,CaiC06,STACS07,STOC09,TAMC09,DyerGP06,mike09,BulatovG05,Homomorphisms,SODA08,BASES,STOC07}.